\documentclass[sigconf]{acmart}

\newif\ifarxiv

\arxivtrue

\usepackage{xspace}
\usepackage{multirow}
\usepackage{comment}
\usepackage{caption}
\usepackage{subcaption}
\usepackage[inline]{enumitem}
\usepackage{amsmath}
\usepackage{amsthm}
\usepackage{url}
\usepackage{pbox}
\usepackage{balance}
\usepackage{wrapfig}
\usepackage{relsize}
\usepackage{tabularx}
\usepackage[noend]{algpseudocode}
\usepackage{algorithm}
\usepackage{xcolor}
\usepackage[T1]{fontenc}
\usepackage[utf8]{inputenc}
\usepackage{mathtools}
\usepackage{bbm}
\usepackage{enumitem}
\setlist{nolistsep}
\usepackage{hhline}
\usepackage{appendix}
\usepackage{multirow}
\makeatletter
\def\maxwidth{ %
  \ifdim\Gin@nat@width>\linewidth
    \linewidth
  \else
    \Gin@nat@width
  \fi
}
\makeatother

\definecolor{fgcolor}{rgb}{0.345, 0.345, 0.345}

\usepackage{framed}
\makeatletter
 {\par\unskip\endMakeFramed%
 \at@end@of@kframe}
\makeatother

\definecolor{shadecolor}{rgb}{.97, .97, .97}
\definecolor{messagecolor}{rgb}{0, 0, 0}
\definecolor{warningcolor}{rgb}{1, 0, 1}
\definecolor{errorcolor}{rgb}{1, 0, 0}
\usepackage{alltt}
\usepackage{hyperref}
\ifarxiv
\fi
\hypersetup{%
  colorlinks=true,      
  linkcolor=blue,       
  citecolor=magenta,    
  filecolor=cyan,       
  urlcolor=blue          
}
\hypersetup{pdfauthor={Name}}
\usepackage[export]{adjustbox}
\usepackage[subtle]{savetrees}
\usepackage{outlines}
\newcommand{\cut}[1]{}

\newcommand{\paragraphb}[1]{\noindent{\bf #1}}

\newcommand{\eg}{e.g., }
\newcommand{\etc}{{etc.}\xspace}
\newcommand{\ie}{i.e., }

\newcommand{\name}{{Mobius}\xspace}

\newcommand{\maas}{{mobility}\xspace}

\newcommand{\qos}{{quality-of-service}\xspace}

\newcommand{\fig}{{Fig.}\xspace}
\newcommand{\app}{{App.}\xspace}
\newcommand{\eqn}[1]{Equation~(#1)}

\newcommand{\im}{{interest map}\xspace}
\newcommand{\ims}{{interest maps}\xspace}

\newcommand{\Ims}{{Interest maps}\xspace}

\newcommand{\ortools}{OR-Tools\xspace}

\newtheorem{theorem}{Theorem}
\newtheorem{lemma}{Lemma}
\newtheorem{definition}{Definition}
\newtheorem{corr}{Corollary}

\setcopyright{rightsretained}

\begin{document}
\acmYear{2021}\copyrightyear{2021}
\acmConference[MobiSys '21]{The 19th Annual International Conference on Mobile Systems, Applications, and Services}{June 24--July 2, 2021}{Virtual, WI, USA}
\acmBooktitle{The 19th Annual International Conference on Mobile Systems, Applications, and Services (MobiSys '21), June 24--July 2, 2021, Virtual, WI, USA}
\acmPrice{}
\acmDOI{10.1145/3458864.3467881}
\acmISBN{978-1-4503-8443-8/21/06}



\title{Throughput-Fairness Tradeoffs in Mobility Platforms}
\author{
    Arjun Balasingam\texorpdfstring{$^\star$}{*},
    Karthik Gopalakrishnan\texorpdfstring{$^\star$}, 
    Radhika Mittal\texorpdfstring{$^\dagger$}{+}, 
    Venkat Arun\texorpdfstring{$^\star$}{*},\texorpdfstring{\\}{}
    Ahmed Saeed\texorpdfstring{$^\star$}{*},
    Mohammad Alizadeh\texorpdfstring{$^\star$}{*},
    Hamsa Balakrishnan\texorpdfstring{$^\star$}{*},
    Hari Balakrishnan\texorpdfstring{$^\star$}{*}
}
\affiliation{
    $^\star$Massachusetts Institute of Technology \hspace{5pt}
    $^\dagger$University of Illinois at Urbana-Champaign
}

\renewcommand{\authors}{Arjun Balasingam,
Karthik Gopalakrishnan,
Radhika Mittal,
Venkat Arun,
Ahmed Saeed,
Mohammad Alizadeh,
Hamsa Balakrishnan,
Hari Balakrishnan}

\ifarxiv
    \settopmatter{printacmref=false}
    \setcopyright{none}
    \renewcommand\footnotetextcopyrightpermission[1]{}
    \pagestyle{fancy}
    \fancyhead[L]{\small Throughput-Fairness Tradeoffs in Mobility Platforms}
    \fancyhead[R]{\small Balasingam, Gopalakrishnan, Mittal, et al.}
    \fancyfoot{}
    \fancyfoot[C]{\thepage}
\fi
\renewcommand{\shortauthors}{Balasingam, Gopalakrishnan, Mittal, et al.}

\begin{abstract}

This paper studies the problem of allocating tasks from different customers to vehicles in mobility platforms, which are used for applications like food and package delivery, ridesharing, and mobile sensing. 
A mobility platform should allocate tasks to vehicles and schedule them in order to optimize both throughput and fairness across customers.
However, existing approaches to scheduling tasks in mobility platforms ignore fairness.

We introduce \name, a system that uses guided optimization to achieve both high throughput and fairness across customers. \name supports spatiotemporally diverse and dynamic customer demands. It provides a principled method to navigate inherent tradeoffs between fairness and throughput caused by shared mobility. Our evaluation demonstrates these properties, along with the versatility and scalability of \name, using traces gathered from ridesharing and aerial sensing applications. Our ridesharing case study shows that \name can schedule more than 16,000 tasks across 40 customers and 200 vehicles in an online manner.

\end{abstract}

\maketitle

\begin{sloppypar}
    \section{Introduction}
\label{sec:intro}

The past decade has seen the rapid proliferation of mobility platforms that use a fleet of mobile vehicles to provide different services. Popular examples include package delivery (UPS, DHL, FedEx, Amazon), food delivery (DoorDash, Grubhub, Uber Eats), and rideshare services (Uber, Lyft). In addition, new types of mobility platforms are emerging, such as drones-as-a-service platforms~\cite{flytos,androne,uav-service,uav-cloud} for deploying different sensing applications on a fleet of drones.

In these mobility platforms, the vehicle fleet of cars, vans, bikes, or drones is a \emph{shared infrastructure}. The platform serves multiple \emph{customers}, with each customer requiring a \emph{set of tasks} to be completed. For instance, each restaurant subscribing to DoorDash is a customer, with several food delivery orders (or tasks) in a city. Similarly, an atmospheric chemist and a traffic analyst might subscribe to a drones-as-a-service platform, each with their own sensing applications to collect air quality measurements and traffic videos, respectively, at several locations in the same urban area. Multiplexing tasks from different customers on the same 
vehicles can increase the efficiency of mobility platforms because vehicles can amortize their travel time by completing co-located tasks (belonging to either the same or different customers) in the same trip.

We study the problem of scheduling spatially distributed tasks from multiple customers on a shared fleet of vehicles. This problem involves (i) assigning tasks to vehicles and (ii) determining the order in which each vehicle must complete its assigned tasks. The constraints are that each vehicle has bounded resources (fuel or battery). While several variants of this scheduling problem 
have been studied, the objective has typically been to complete as many tasks as possible in bounded time, or to maximize aggregate throughput (task completion rate)~\cite{vrp,vrp-applications}. 

We identify a second---equally important---scheduling requirement, which has emerged in today's customer-centric mobility platforms: \emph{fairness} of customer throughput to ensure that tasks from different customers are fulfilled at similar rates.\footnote{The method we develop also applies to weighted fairness.} For example, in food delivery, the platform should serve restaurants equitably, even if it means spending time or resources on restaurants with patrons far from the current location of the vehicles. A ridesharing platform should ensure that riders from different neighborhoods are served equitably, which ridesharing platforms today do not handle well, a phenomenon known as ``destination discrimination''~\cite{middleton2018discrimination, uber-discrimination, uber-dd}.

We seek an online scheduler for mobility platforms that achieves both high throughput and fairness. A standard approach to achieving these goals is to track the resource usage and work done on behalf of different users in a fine-grained way and equalize resource consumption across users.
Such fine-grained accounting and attribution is difficult with shared mobility: the resource used is a moving vehicle traveling toward its next task, but making that trip has a knock-on benefit, not only for the next task served, but for subsequent ones as well. However, the benefit of a specific trip is not equal across the subsequent tasks. Although it may be possible to develop a fair scheduler that achieves high throughput using fine-grained accounting and attribution, it is likely to be complex.

We turn, instead, to an approach that has been used in both societal and computing systems: optimization, which may be viewed as a search through a set of feasible schedules to maximize a utility function. In our case, we can establish such a function, optimize it using both the task assignment and path selection, and then route vehicles accordingly. 

In a typical mobility problem, the planning time frame for optimization could be between 30 minutes and several hours, involving hundreds of vehicles, dozens of customers, and tens of thousands of tasks. The scale of this problem pushes the limits of state-of-the-art vehicle routing solvers~\cite{vrp-largescale}. Moreover, fairness objectives lead to nonlinear utility functions, which make the optimization much more challenging. As a benchmark, optimizing the routes for 3 vehicles and 17 tasks over 1 hour, using the CPLEX solver~\cite{cplex} with a nonlinear objective function, takes over 10 hours~\cite{molina2014multi}.

To address these problems, a natural approach is to divide the desired time duration into shorter rounds, and then run the utility optimization. When we do this, something interesting emerges in mobility settings: the space of feasible solutions---each solution being an achievable set of rates for the customers---often {\em collapses into a rather small and disturbingly suboptimal set!} These feasible solutions are either  fair but with dismal throughput, or with excellent throughput but starving several customers. 

A simple example helps see why this happens. Consider a map with three areas, $A_1$, $A_2$, $A_3$, each distant from the others. There are several tasks in each area: in $A_1$, all the tasks are for customer $C_1$; in $A_2$, all the tasks are for customer $C_2$, and in $A_3$, all the tasks are for two other customers, $C_3$ and $C_4$. Suppose that there are two vehicles. Over 
a duration of a few minutes, we could either have the two vehicles focus on only two areas, achieving high throughput but ignoring the third area and reducing fairness, or, we could have them move between areas after each task to ensure fairness, but waste a lot of time traveling, degrading throughput. It is not possible here to achieve both throughput and fairness {\em over a short timescale}.
Yet, over a long time duration, we can swap vehicles between regions to amortize the movement costs. This shows that planning over a longer timescale permits feasible schedules that are better than what a shorter timescale would permit. 

Our contribution, \emph{\name{}}, divides the desired time duration into rounds, and produces the feasible set of allocations for that round using a standard optimizer. \name guides the optimizer toward a solution that is not in the feasible set for one round but can be achieved over multiple rounds. This guiding is done by aiming for an objective that maximizes a weighted linear sum of customer rates in each round. The weights are adjusted dynamically based on the long-term rates achieved for each customer thus far. The result is a practical system that achieves high throughput and fairness over multiple rounds. This approach of achieving long-term fairness by setting appropriate weights across rounds allows us to use off-the-shelf solvers for the weighted Vehicle Routing Problem (VRP) for path planning in each round.
Importantly, this design allows \name to optimize for fairness in the context of any VRP formulation, making this work complementary to the vast body of prior work on vehicle routing algorithms~\cite{vrp-applications,vrp-maxtpt,vrp-stochastic,pctsp}.

Scheduling over multiple rounds also allows \name to handle tasks that arrive dynamically or expire before being done. Moreover, \name supports a tunable level of fairness modeled by $\alpha$-fair utility functions~\cite{prop-fair}, which generalize the familiar notions of max-min and proportional fairness.

We have implemented \name and evaluated it via extensive trace-driven emulation experiments in two real-world settings: (i) a ridesharing service, based on real Lyft ride request data gathered over a day, ensuring fair \qos to different neighborhoods in Manhattan; and (ii) urban sensing using drones for measuring traffic congestion, parking lot occupancy, cellular throughput, and air quality. We find that:
\begin{enumerate}[label=\arabic*.]
\item Relative to a scheduler that maximizes only throughput, \name compromises only 10\% of platform throughput in order to enforce max-min fairness.
\item Compared to dedicating vehicles to customers, \name improves vehicle utilization by 30-50\% by intelligently sharing vehicles amongst customers.
\item \name can compute fair online schedules at a city scale, involving 40 customers, 200 vehicles, and over 16,000 tasks.
\end{enumerate}

    \section{Problem Setup}
\label{sec:motivation}

\begin{figure*}
	\centering
	\includegraphics[width=1\textwidth]{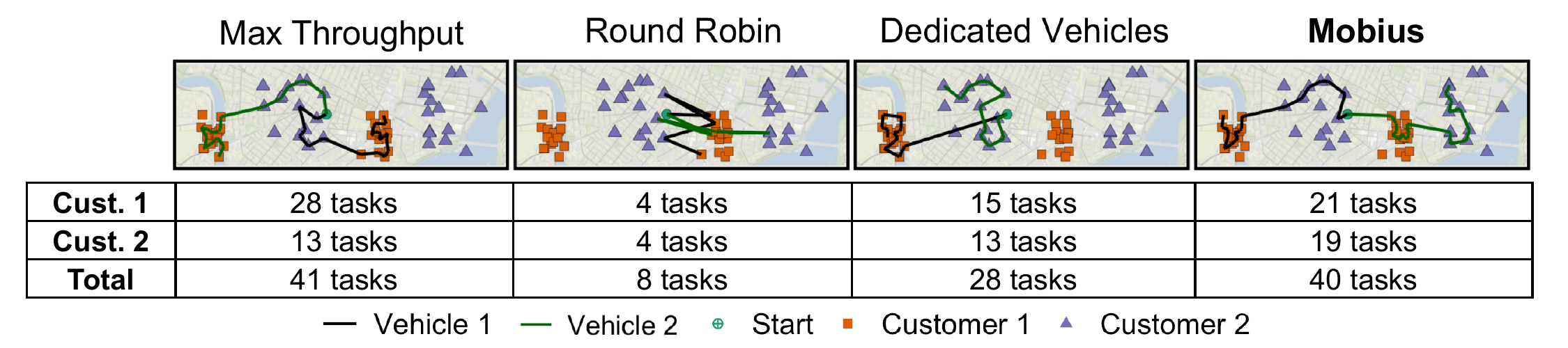}
	\vspace{-15pt}
	\caption{An example with two customers, two vehicles, and a 6-minute planning horizon. \name computes a schedule that (i) achieves a similar total throughput to that of the max throughput schedule, and (ii) preserves the customer-level fairness achieved by the round-robin and dedicated schedules.}
    \label{fig:motivation}
 \end{figure*}

Every customer subscribing to a \maas platform submits several requests over time. Each request specifies a task (\eg gather sensor data or deliver package) and a corresponding location. 
The platform schedules trips for each vehicle over multiple rounds.
It takes into account any changes in a customer's requirements (in the form of new task requests or expiration of older unfulfilled tasks) at the beginning of each round. We say that a customer has a backlog if they have more tasks than can be completed by all available resources within the allocated time.   
For simplicity of exposition, we assume each customer is backlogged (our evaluation in \S\ref{sec:eval} relaxes this assumption).

Let $K$ be the set of customers, 
and $T_k(\tau)$ be the set of tasks requested by customer $k$ during a scheduling round $\tau$. 
We denote $x_k(\tau)$ as the throughput achieved for customer $k$ in scheduling round $\tau$, \ie the total number of tasks in $T_k(\tau)$  that are fulfilled 
divided by the round duration.

We denote $\overline{x}_k(t)$ as the long-term throughput for each customer $k$, after $t$ scheduling rounds, \ie $\overline{x}_k(t) = \frac{1}{t} \sum_{\tau = 1}^t x_k(\tau)$ if rounds are of equal duration.
A good scheduling algorithm should achieve the following objectives:
\begin{itemize}
\item \textbf{Platform Throughput.} Maximize the total long-term throughput after round $t$, \ie $\sum_{k \in K} \overline{x}_k(t)$. 
\item \textbf{Customer Fairness.} For any two customers $k_1, k_2 \in K$ with backlogged tasks, ensure $\overline{x}_{k_1}(t) = \overline{x}_{k_2}(t)$. 
\end{itemize}
Equalizing long-term per-customer throughputs $\overline{x}_k(t)$ provides a desirable measure of fairness for many mobility platforms: higher per-customer throughputs correlate with other performance metrics, such as lower task latency and higher revenue. Our evaluation (\S\ref{sec:eval}) quantifies the impact of optimizing for a fair allocation of throughputs on other platform-specific quality-of-service metrics.

Prior algorithms for scheduling tasks on a shared fleet of vehicles have focused on the VRP, \ie only considered maximizing platform throughput~\cite{vrp,vrp-applications}.  Achieving per-customer fairness introduces three new challenges:

\vspace{3pt}
\paragraphb{Challenge \#1: Attributing vehicle time to customers.} 
Vehicle time and capacity are scarce. 
Consider the example in \fig~\ref{fig:motivation}, with two customers and two vehicles; customer 1 has two densely-packed clusters of tasks, while customer 2 has two dispersed clusters of tasks. We show schedules and tasks fulfilled by \name and three other policies: (i) maximizing throughput, (ii) dedicating a vehicle per customer, and (iii) alternating round-robin between customer tasks. Notice that, to the left of the depot (center of the map), customer 2's tasks can be picked up on the way to customer 1's tasks. Thus, multiplexing both customers' tasks on the same vehicle is more desirable than dedicating a vehicle per customer, because it amortizes resources to serve both customers. However, sharing vehicles amongst customers complicates our ability to reason about fairness, because the travel time between the tasks of different customers cannot be attributed easily to each one. 

\begin{figure}[t]
  \centering
  \includegraphics[scale=0.63]{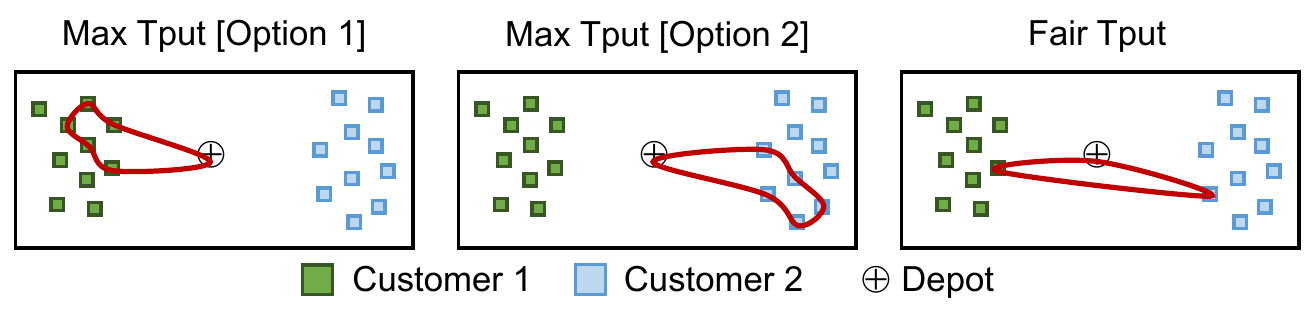}\vspace{-10pt}
  \caption{Imposing fairness at short timescales (\eg one round trip) degrades throughput. Executing Options 1 and 2 provides fairness at longer timescales and leads to greater total throughput.
  }
  \vspace{-10pt}
  \label{fig:motivation:toy}
\end{figure}

\vspace{3pt}
\paragraphb{Challenge \#2: Timescale of fairness.}
\fig~\ref{fig:motivation:toy} shows two customers and one vehicle that must return home to refuel. A high-throughput schedule would dedicate the vehicle to one of the customers. By contrast, a fair schedule would require the vehicle to round-robin customer tasks, achieving low throughput due to travel. Over a longer time duration, however, we can execute two max-throughput schedules (Options 1 and 2) to achieve both fairness and high throughput.

\vspace{3pt}
\paragraphb{Challenge \#3: Spatiotemporal diversity of tasks.}
In \fig~\ref{fig:motivation}, the two customers' tasks have different spatial densities.
The high-throughput schedule favors customer 1. A max-min fair schedule should, by contrast, ensure that customer 2 gets its fair share of the throughput, even if it comes at the cost of higher travel time and lower platform throughput. Striking the right balance between fairly serving a customer with more dispersed tasks and reducing extra travel time is a non-trivial problem.

Customer tasks may also \emph{vary with time}. For example, a food delivery service might receive new requests from restaurants, or an atmospheric scientist may want to update sensing locations that they submitted to a drone service provider based on prior observations. The mobility platform must handle the dynamic arrival and expiration of tasks.

    \section{Overview}
\label{sec:overview}

\begin{figure}[t]
  \centering
  \includegraphics[scale=0.6]{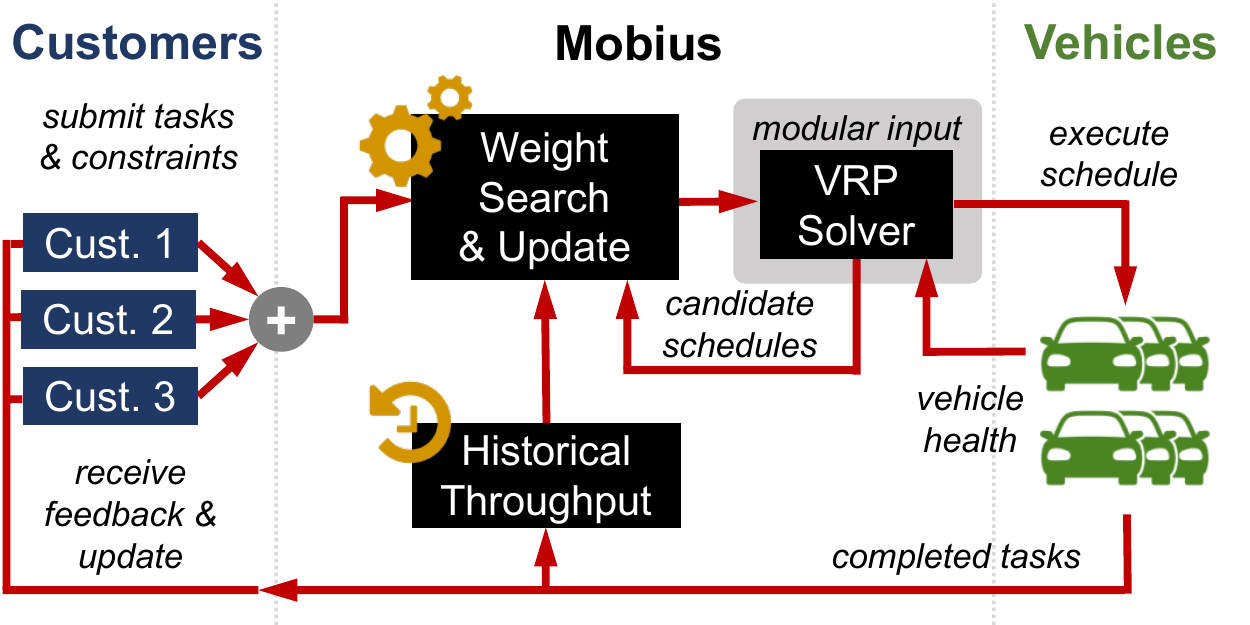}
  \vspace{-10pt}
  \caption{In each round, \name uses a VRP solver to compute a schedule that maximizes a weighted sum of throughputs, and automatically adjusts the weights across rounds to improve fairness.}
  \vspace{-10pt}
  \label{fig:overview}
\end{figure}

Any resource-constrained system exhibits an inherent tradeoff between throughput and fairness. In the best case, the most fair schedule would also have the highest throughput; however, due to the challenges described in \S\ref{sec:motivation}, it is impossible to realize this goal in many mobility settings. \name instead strives for customer fairness with the best possible platform throughput; its approach is to trade some short-term fairness for a boost in throughput, while improving fairness over a longer timescale.

In each round $\tau$, \name uses a VRP solver to maximize a weighted sum of customer throughputs $x_k(\tau)$.\footnote{We formally define the VRP in \S\ref{sec:design}.} 
\name sets the weights in each round to find a high throughput schedule that is approximately fair in that round. By accounting for the long-term throughputs $\overline{x}_k(t)$ delivered to each customer $k$ in prior rounds, it is able to equalize $\overline{x}_k(t)$ over multiple rounds. We formalize this notion of balancing high throughput with fairness in \S\ref{sec:frontier}.
\name uses an iterative search algorithm requiring multiple invocations of a VRP solver to find a schedule that strikes the appropriate balance. 

Our approach of trading off short-term fairness for throughput and longer-term fairness 
raises a natural question: why not directly schedule over a longer time horizon, rather than dividing the scheduling problem into rounds? Scheduling in rounds is desirable for several reasons: (i) their duration can correlate with the fuel or battery constraints of the vehicles, (ii) it provides a target timescale at which \name strives to provide fairness, (iii) shorter timescales make the NP-hard VRP problem more tractable to solve, and (iv) it enables \name to adapt to temporal variations in customer demand that are captured at the beginning of each round.

\fig~\ref{fig:overview} shows the architecture of \name. In each round, customers update their task requests. \name then computes the best weights, generates a schedule, and dispatches the vehicles. At the end of the round, \name updates each customer's throughput, $\overline{x}_k(t)$, and uses this information to select weights in the next round. 

    \section{Balancing Throughput \& Fairness}
\label{sec:frontier}

\begin{figure*}[t]
	\centering
    \begin{subfigure}[t]{0.23\textwidth}
		\centering
		\includegraphics[scale=0.44]{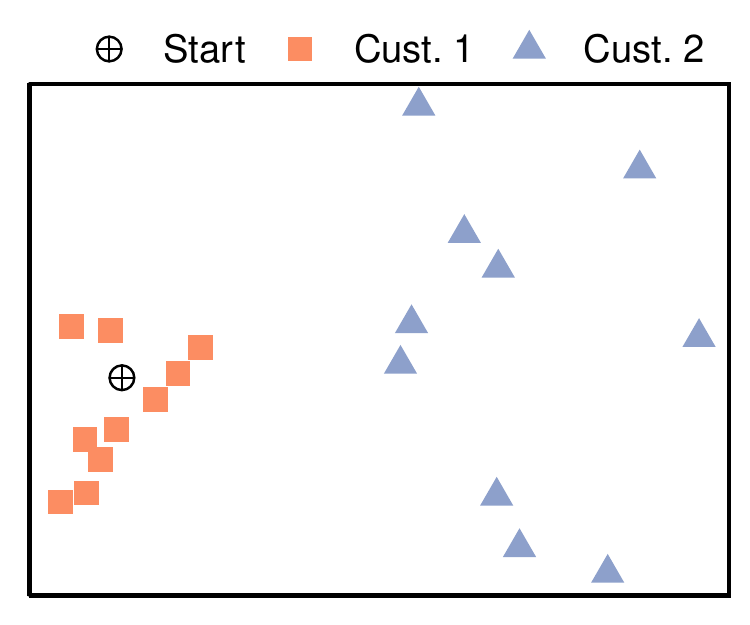}
		\caption{Map. Two vehicles start at $\oplus$.}
		\label{fig:frontier:bf-map}
	\end{subfigure}
    \begin{subfigure}[t]{0.25\textwidth}
		\centering
		\includegraphics[scale=0.44]{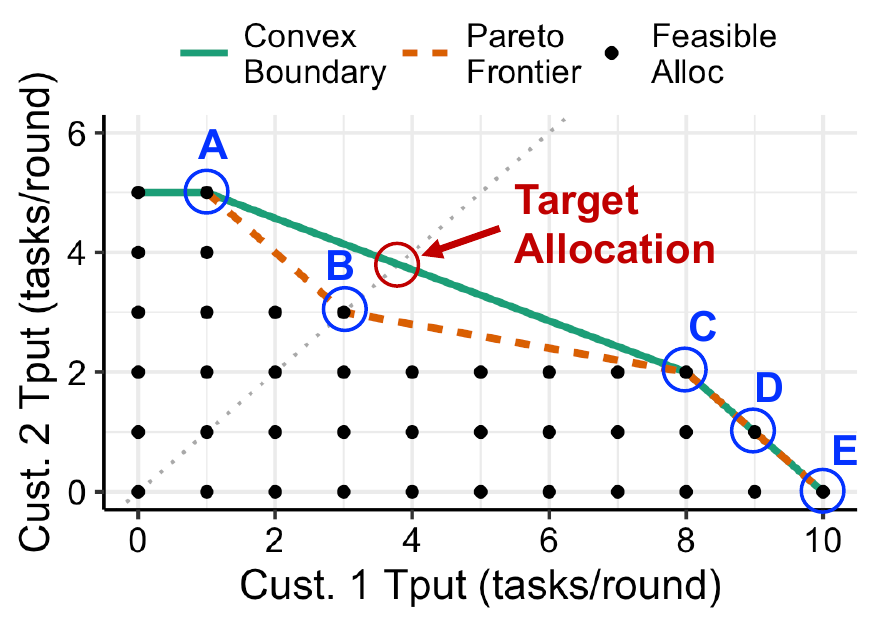}
		\caption{Feasible throughputs in 1 round.}
		\label{fig:frontier:bf-frontier}
	\end{subfigure}
      \begin{subfigure}[t]{0.25\textwidth}
      \centering
      \includegraphics[scale=0.44]{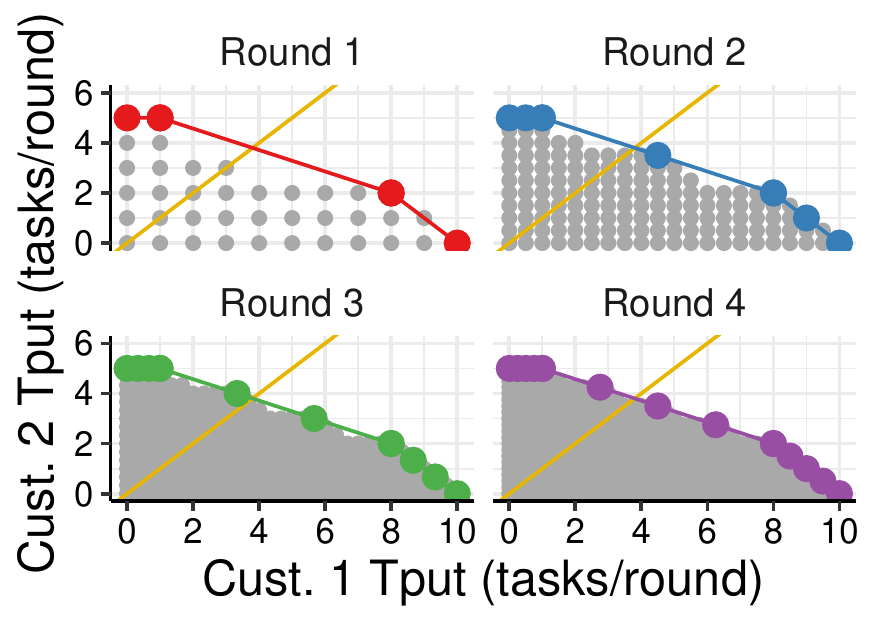}
      \caption{Feasible throughputs over 4 rounds.}
      \label{fig:frontier:bf-rate}
      \end{subfigure}
       \begin{subfigure}[t]{0.25\textwidth}
      \centering
		\includegraphics[scale=0.44]{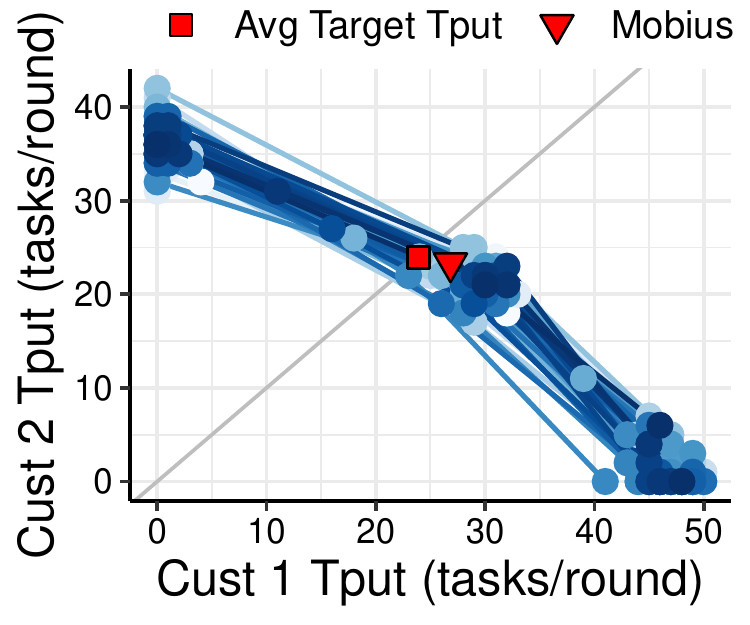}
		\caption{Convex boundary dynamics.}
		\label{fig:frontier:dynamic-hull}
      \end{subfigure}
	\vspace{-10pt}
	\caption{Visualizing feasible allocations of throughput for a small problem with two customers and two vehicles. Allocations on the convex boundary trade short-term fairness for throughput. The convex boundary becomes denser over time, making the target allocation achievable.}
	\vspace{-10pt}
	\label{fig:frontier:model}
\end{figure*}

We now provide the intuition behind our approach for balancing throughput and fairness using the example shown in \fig~\ref{fig:frontier:model}. There are two customers, each requesting tasks from distributions shown on the map in \fig~\ref{fig:frontier:bf-map}. We have two vehicles, each starting at $\oplus$. For simplicity, in \S\ref{sec:frontier:convex}, we consider planning schedules in 10-minute rounds, where the vehicles return to their start locations after 10 minutes. We renew all tasks at the beginning of each round trip. Then, in \S\ref{sec:frontier:dynamics}, we explain how \name generalizes to dynamic settings where customer tasks change with time, and vehicles do not need to return to their starting locations.

\subsection{Scheduling on the Convex Boundary}
\label{sec:frontier:convex}
\paragraphb{Feasible allocations.}
We first consider the set of schedules that are feasible within the time constraint. \fig~\ref{fig:frontier:bf-frontier} shows the tradeoff between throughput and fairness amongst these feasible schedules. Each dot represents an allocation produced by a feasible schedule; the coordinates of the dot indicate the throughputs of the respective customers. We generate the schedules by solving the VRP for each possible subset of customer tasks.\footnote{The VRP is NP-hard (\S\ref{sec:design}), but because the input size is small for this example, we use Gurobi~\cite{gurobi} to compute optimal schedules.} We also indicate the $y=x$ line (dotted gray), which corresponds to fair allocations that give equal throughput to each customer. Note that in this example both vehicles can more easily service Customer 1. Hence, an allocation that maximizes total throughput without regard to fairness (labeled $C$) favors Customer 1.

\vspace{3pt}
\paragraphb{Pareto frontier of feasible allocations.} 
The Pareto frontier over all feasible allocations is denoted by the dashed orange line, containing $A$, $B$, $C$, $D$, and $E$.
If an allocation on the Pareto Frontier achieves throughputs of $x_1$ and $x_2$ for Customers 1 and 2 respectively, there exists no feasible allocation $(\hat{x}_1, \hat{x}_2)$ such that $\hat{x}_1 > x_1$ and $\hat{x}_2 > x_2$.
The allocation that maximizes total throughput will always lie on the Pareto frontier. An allocation on the Pareto frontier is strictly superior, and therefore more desirable than other feasible allocations. So which allocation on the Pareto frontier do we pick? A strictly fair allocation lies at the point where the Pareto frontier intersects the $y=x$ line (labeled $B$ in \fig~\ref{fig:frontier:bf-frontier}). However, allocation $B$ has low total throughput, because the vehicles spend a significant part of the 10 minutes traveling between task clusters.

\vspace{3pt}
\paragraphb{Convex boundary of the Pareto frontier.} 
To capture the subset of allocations that do not significantly compromise throughput, we use the \emph{convex boundary} of all feasible allocations, denoted by the turquoise line in \fig~\ref{fig:frontier:bf-frontier}. 
The convex boundary is the smallest polygon around the feasible set such that no vertex bends inward~\cite{boyd-convexset}, and the \emph{corner points} are the vertices determining this polygon. 
The \emph{target allocation} is the point where the $y=x$ lines intersects the boundary (shown in red). It has high throughput and is fair, but it may not be feasible (as in this example). Is it still possible to achieve the target throughput in such cases?

\vspace{3pt}
\paragraphb{Scheduling over multiple rounds.}
Our key insight is that it is possible to achieve the target allocation over multiple rounds of scheduling by selecting different feasible allocations on the convex boundary in each round. 
In a given round, \name chooses the feasible allocation on the convex boundary that best achieves our fairness criteria. In our example, it chooses allocation $A$ in its first round. By choosing allocation $A$ over allocation $B$ (which achieves equal throughput), \name compromises on short-term fairness for a boost in throughput. However, as we discuss next, it compensates for this choice in subsequent rounds. Notice that if \name instead chooses $B$, it would not be able to recover from the resulting loss in throughput.

As we compute a 10-minute schedule for each round, the set of feasible allocations expands;
this allows \name to compensate for any prior deviation in fairness. \fig~\ref{fig:frontier:bf-rate} illustrates how the feasible set evolves over several 10-minute rounds of planning.
The feasible allocations (denoted by gray dots) possible after round $T$ are derived from the cumulative set of tasks completed in $T$ rounds.
Notice that over the four rounds, the set of feasible allocations becomes denser, and the Pareto frontier approaches the convex boundary.
Thus, the target allocation (\ie the allocation on the convex boundary that coincides with the $y=x$ line) becomes feasible.

In summary, the key insights driving the design of \name are: (i) the convex boundary describes a set of allocations that trade off short-term fairness for a boost in throughput, and (ii) the Pareto frontier approaches the convex boundary over multiple rounds of planning, making it possible to correct for unfairness over a slightly longer timescale.

\subsection{Scheduling in Dynamic Environments}
\label{sec:frontier:dynamics}
In practice, environments are more dynamic: customer tasks may not recur at the same locations, and vehicles need not return to their start locations regularly. Thus, the convex boundary may not be identical in each round. However, in practice, because (i) vehicles move continuously over space and (ii) customer tasks tend to observe spatial locality, the convex boundary does not change drastically over time.

To illustrate this, we extend the example in \fig~\ref{fig:frontier:model}, by creating a map with the same densities as in \fig~\ref{fig:frontier:bf-map}, but with 50 tasks per customer. To simulate dynamics, we create a new task for each customer every 3 minutes at a location chosen uniformly at random within a bounding box.  
We still consider two vehicles starting at the same location (\ie in the middle of customer 1's cluster) and plan in 10-minute rounds. We eliminate the return-to-home constraint. In order to adapt to the customers' changing tasks, we compute new 10-minute schedules every 1 minute (\ie 10-minute rounds slide in time by 1 minute). We run this simulation for 60 minutes.

In order to understand how these dynamics impact the convex boundary as we plan iteratively, we show in \fig~\ref{fig:frontier:dynamic-hull} the convex boundary of 10-minute schedules at each 1-minute replanning interval. Notice that the convex boundaries hover around a narrow band, indicating that we can still track the target throughput reliably. The red square marks the value of the average target throughput across all 50 convex boundaries; 
we also mark the throughput achieved by \name's scheduling algorithm (\S\ref{sec:design}).

In addition to the convex boundary remaining relatively stable from one timestep to the next, this method of replanning at much quicker intervals (\eg 1 minute) than the round duration (\eg 10 minutes) makes \name resilient to uncertainty in the environment.\footnote{\S\ref{sec:eval} further evaluates the effectiveness of \name's algorithm for dynamic, real-world customer demand.} For instance, \name can react to streaming requests in a punctual manner, and can also incorporate requests that are unfulfilled due to unexpected delays (\eg road traffic or wind). Moreover, since \name uses a VRP solver as a building block to compute its schedule (\S\ref{sec:overview}), it can also leverage algorithms that solve the stochastic VRP~\cite{vrp-stochastic}, where requests arrive and disappear probabilistically.

\subsection{Visualizing Routes Scheduled by \name}
\label{sec:frontier:viz}
\begin{figure}
  \centering
  \includegraphics[scale=0.65]{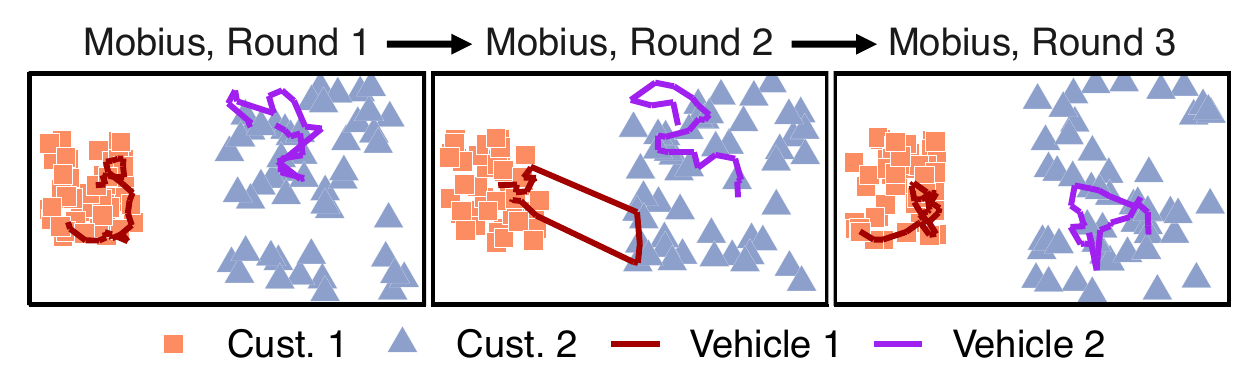}
  \vspace{-10pt}
  \caption{The difference in spatial density of tasks leads to short-term unfairness (Rounds 1 and 3). \name compensates for this by directing more resources to the underserved customer (Round 2).}
  \vspace{-10pt}
  \label{fig:frontier:routes}
\end{figure}

To illustrate how \name converges to fair per-customer allocations without significantly degrading platform throughput,  in \fig~\ref{fig:frontier:routes} we show 3 consecutive 10-minute round schedules computed by \name, for the dynamic example in \fig~\ref{fig:frontier:dynamic-hull}. In Rounds 1 and 3, we observe that \name decides to dedicate one vehicle to each customer in order to give them both sufficiently high throughput; here, customer 2 receives lower throughput because its tasks are more dispersed. However, in Round 2, \name compensates for this short-term unfairness by scheduling an additional vehicle to customer 2, while also collecting a few tasks for customer 1 in the outbound trip. 

    \section{\name Scheduling Algorithm}
\label{sec:design}

Based on the insights in \S\ref{sec:frontier}, we design \name to compute a schedule on the convex boundary in each round,
such that the long-term throughputs $\overline{x}_k(t)$ approach the target allocation. \name works in two steps:
\begin{enumerate}
\item In each round, \name finds the \emph{support allocations}, which we define as the corner points on the convex boundary of the current round, near the target allocation (\S\ref{sec:design:hull}). For example, in \fig~\ref{fig:frontier:bf-frontier}, \name would find support allocations $A$ and $C$.
\item Amongst the support allocations found in step (1), \name selects the one that \emph{steers the long-term throughputs} $\overline{x}_k(t)$ toward the target allocation (\S\ref{sec:design:rounds}).
\end{enumerate}
In this section, we present \name in the context of strict fairness (\ie $\overline{x}_k(t)$ must lie along the $y=x$ line).
\S\ref{sec:design:opt} provides a theoretical analysis of \name's optimality under simplifying assumptions, and \S\ref{sec:design:impl} describes our implementation. In \S\ref{sec:alpha}, we extend \name's formulation to work with a class of fairness objectives. 

\subsection{Finding Support Allocations}
\label{sec:design:hull}

Since the convex boundary of the Pareto frontier is equivalent to the convex boundary of the feasible set of schedules, a naive way to find the support allocations is to compute the Pareto frontier, take its convex boundary, and then identify the support allocations near the target allocation. However, computing the Pareto frontier is computationally expensive because it requires invoking an NP-hard solver an exponential number of times in the number of tasks. \name uses a VRP solver as a building block to find a subset of the corner points of the convex boundary around the target allocation.

The VRP involves computing a path
$\mathcal{P}_v$ for each vehicle $v$, defined as an ordered list of tasks from the set of all tasks $\{T_k(\tau) \; \mid k \in K\}$, such that the time to complete $\mathcal{P}_v$ does not exceed the total time budget $B$ for a round. VRP solvers maximize the platform throughput without regard to fairness. 

We capture different priorities amongst customer tasks by assigning a weight $w_k$ to each customer $k$'s tasks. Let $\mathbf{x} \in \mathbb{R}^{|K|}$ represent a throughput vector, where $x_k$ is the throughput for customer $k$, and let $\mathbf{w} \in \mathbb{R}^{|K|}$ represent a weight vector, with a weight $w_k$ for each customer $k$.\footnote{$\mathbf{x}$ and $\mathbf{w}$ vary with each round $\tau$. We drop the round index $\tau$ whenever there is no ambiguity about the current round.}

The weighted VRP seeks to maximize the total \emph{weighted throughput} of the system, 
where each task is allowed a weight. We can describe this as a mixed-integer linear program:
\begin{align}
& \underset{\mathcal{P}_v, \; \forall v \in V}{\text{argmax}} & \sum_{k \in K} w_k x_k  \;\; = \;\; &\underset{\mathcal{P}_v, \; \forall v \in V}{\text{argmax}} \;\; \mathbf{w}^\intercal\mathbf{x} & \label{eqn:vrp:obj} \\
& \textrm{s.t.} & c(\mathcal{P}_v) \leq B & & \forall v \in V \label{eqn:vrp:budget} \\
& & \mathcal{P}_v \textrm{ is a valid path} & & \forall v \in V, \label{eqn:vrp:constraints}
\end{align}
where $c(\cdot)$ specifies the time to complete a path. \eqn{\ref{eqn:vrp:budget}} enforces that, for each vehicle, the time to execute the selected path does not exceed the budget. \eqn{\ref{eqn:vrp:constraints}} captures constraints that are specific to the vehicles (\eg if vehicles must return to home at the end of each round) and customers (\eg if tasks are only valid during specific windows during the scheduling horizon). The weighted VRP (also called the prize-collecting VRP) is NP-hard, but there are several known algorithms with optimality bounds~\cite{pctsp, vrp}.

\vspace{3pt}
\paragraphb{Using weights to find the corner points.}
We can adjust the weight vector $\mathbf{w}$ in order to capture a bias toward a particular customer; $\mathbf{w}$ describes a direction in the customer throughput space, reflecting that bias. \fig~\ref{fig:design:search} visualizes $\mathbf{w}$ in a 2-D customer throughput space. A solver optimizing for \eqn{\ref{eqn:vrp:obj}} searches for the schedule with the highest throughput in the direction of $\mathbf{w}$~\cite{boyd-lp}, thus requiring the schedule to lie on the convex boundary. For example, $\mathbf{w_1} = (1,0)$ finds the schedule on the convex boundary that prioritizes customer 1 (\ie along the $x$-axis), and $\mathbf{w_2} = (0,1)$ finds a schedule that prioritizes customer 2 (\ie along the $y$-axis). 

\begin{figure}
	\centering
    \begin{subfigure}[t]{0.23\textwidth}
		\centering
		\includegraphics[scale=0.5]{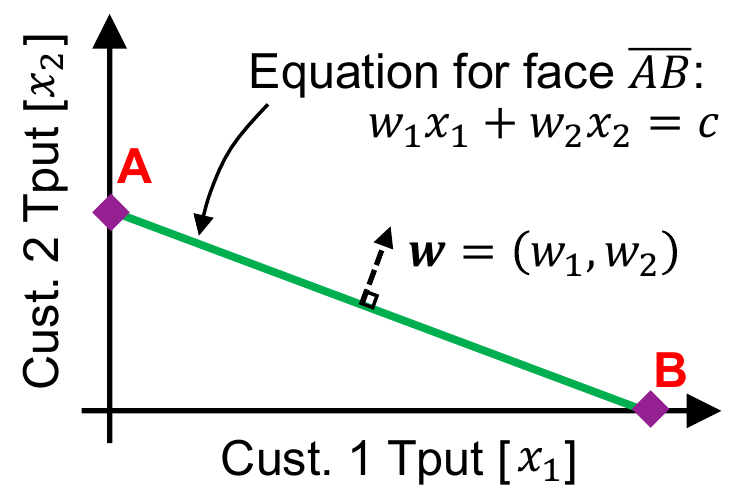}
		\caption{Computing $\mathbf{w}$ for a face.}
		\label{fig:design:search}
	\end{subfigure}
    \begin{subfigure}[t]{0.23\textwidth}
		\centering
		\includegraphics[scale=0.5]{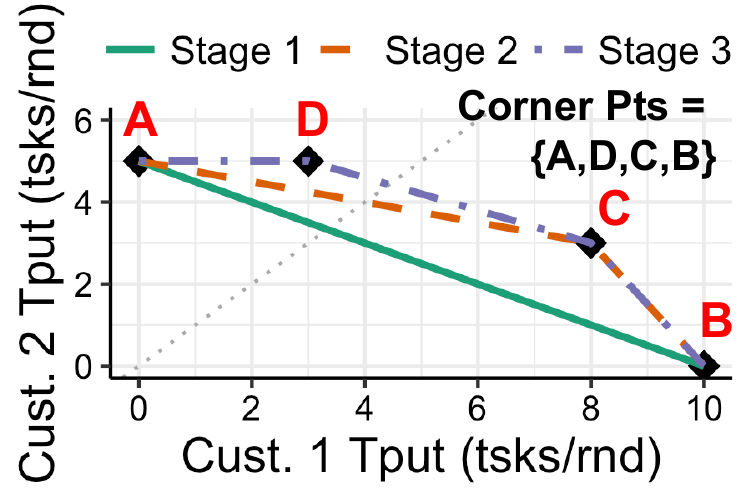}
		\caption{Finding support allocations.}
		\label{fig:design:stages}
	\end{subfigure}
  \vspace{-10pt}
  \caption{Using a blackbox VRP solver as a building block, \name runs an iterative search algorithm to find the support allocations.}
  \vspace{-10pt}
  \label{fig:design:trace}
\end{figure}

\vspace{3pt}
\paragraphb{Searching on the convex boundary.}
Recall that, for strict fairness, the target allocation is the point where the $y=x$ line intersects the convex boundary for the current round (\S\ref{sec:frontier}). At the start of the search, \name does not yet know the convex boundary, so it cannot know the target allocation. 
To find allocations on the convex boundary, \name employs an iterative search algorithm, analogous to binary search; in each stage, it tries to find a new allocation on the convex boundary in the direction of the $y=x$ line.
\name begins the search with allocations along the customer axes. For two customers, it begins with weights $\mathbf{w_1}$ and $\mathbf{w_2}$ above, which gives two allocations on the convex boundary. 
In each stage of the search, \name computes a new weight vector, using allocations found on the convex boundary in the previous stage, in order to find a new allocation on the convex boundary. It terminates when no new allocation can be found. By searching in the right direction, \name only needs to compute a subset of corner points on the convex boundary.

To better illustrate the algorithm, consider the example in \fig~\ref{fig:design:stages}, with 2 customers.
\name starts the search by looking at the two extreme points on the customer 1 ($x_1$) and customer 2 ($x_2$) axes, which correspond to prioritizing all vehicles for either customer. So in stage 1, \name computes these schedules, using the weight vectors $\mathbf{w_1} = (1,0)$ and $\mathbf{w_2} = (0,1)$, which give the allocations  $A$ and $B$, respectively, in \fig~\ref{fig:design:stages}. After stage 1, $\{A, B\}$ is the current set of corner points determining the convex boundary.

In the next stage, \name computes a new weight $\mathbf{w}$ to continue the search in the direction normal to $\overline{AB}$ (\fig~\ref{fig:design:search}). 
Let the equation for the face $\overline{AB}$ be $w_1x_1 + w_2x_2 = c$, where $w_1$, $w_2$, and $c$ can be derived
using the known solutions on the line, $A$ and $B$. 
So, by invoking the VRP solver (\eqn{\ref{eqn:vrp:obj}}) with $\mathbf{w} = (w_1, w_2)$, we try to find a schedule on the convex boundary, with the highest throughput in the direction normal to $\overline{AB}$. Let $\hat{x}_1$ and $\hat{x}_2$ be the throughputs for the schedule computed with weight $\mathbf{w}$. If $(\hat{x}_1, \hat{x}_2)$ lies above this line, \ie $w_1 \hat{x}_1 + w_2 \hat{x}_2 > c$, then the point $(\hat{x}_1, \hat{x}_2)$ is a valid extension to the convex boundary. In this example, \name finds a new allocation $C$; so, the new set of corner points is $\{A, C, B\}$.

Notice that this extension in stage 2 creates two new faces on the convex boundary, $\overline{AC}$ and $\overline{CB}$. But, the $y=x$ line only passes through $\overline{AC}$. So, in stage 3, \name continues the search, extending $\overline{AC}$ by the computing the weights as described above (normal to $\overline{AC}$), and discovers a new allocation $D$. Finally, \name tries to extend the face $\overline{DC}$ because it intersects the $y=x$ line. It finds no valid extension, and so, it terminates its search on the face $\overline{DC}$, and returns the support allocations $D$ and $C$.

\vspace{3pt}
\paragraphb{Generalizing to more customers.}
\name computes a weight for each customer $k \in K$, \ie $\mathbf{w} \in \mathbb{R}^{|K|}$. Faces on the convex boundary become $|K|$-dimensional hyperplanes, described by the equation $\sum_{k \in K} w_k x_k = c$. \name solves for $\mathbf{w}$ using the $|K|$ allocations that define each face, and finds $|K|$ support allocations at the end of the search. 
Recall from the example in \S\ref{sec:design:hull} that each stage produced 2 new faces and that \name only continued the search by extending 1 face. With $|K|$ customers, even with $|K|$ new faces after each stage, \name only invokes the VRP solver once to continue the search. A naive algorithm, by contrast, would require $|K|$ calls to the VRP solver in each stage. Thus \name scales easily with more customers by pruning the search space efficiently.

\subsection{Scheduling Over Rounds}
\label{sec:design:rounds}

In each round, \name finds $|K|$ support allocations, which determine the face of the convex boundary that contains the target allocation. It then selects a support allocation among these $|K|$ such that the per-customer long-term throughputs $\overline{x}_k(t)$ approach the target throughput. By tracking $\overline{x}_k(t)$ over many rounds, \name can select allocations that compensate for any unwanted bias introduced to some customer in a prior round.

Mathematically, to choose a schedule in round $t$, \name considers the effect of each support allocation $\mathbf{x}(t)$ on the average throughput $\overline{\mathbf{x}}(t+1)$. The average throughput is defined for each customer $k$ as $\overline{x}_k(t+1) = \gamma_t x_k(t) + (1-\gamma_t) \overline{x}_k(t),$ where $\gamma_t = 1/(t+1)$. \name chooses $\mathbf{x}(t)$ such that $\overline{\mathbf{x}}(t+1)$ is closest to the $y=x$ line (in Euclidean distance).

\begin{figure}
	\centering
    \begin{subfigure}[t]{0.23\textwidth}
		\centering
		\includegraphics[scale=0.5]{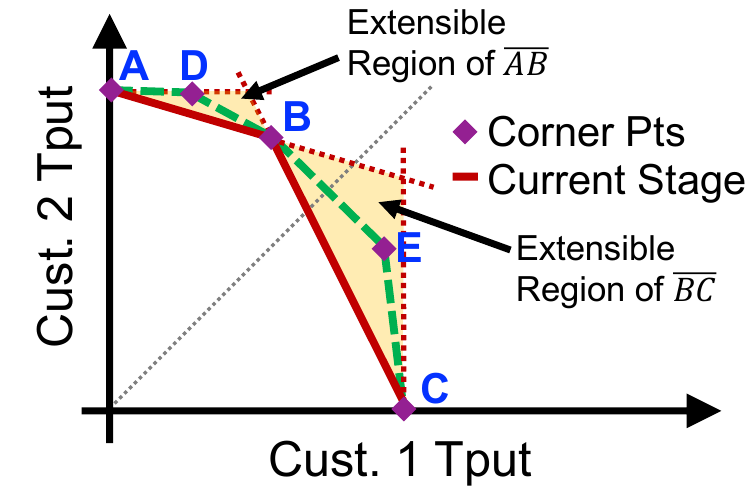}
		\caption{Extensible region of face $\overline{BC}$.}
		\label{fig:design:ext}
	\end{subfigure}
    \begin{subfigure}[t]{0.23\textwidth}
		\centering
		\includegraphics[scale=0.5]{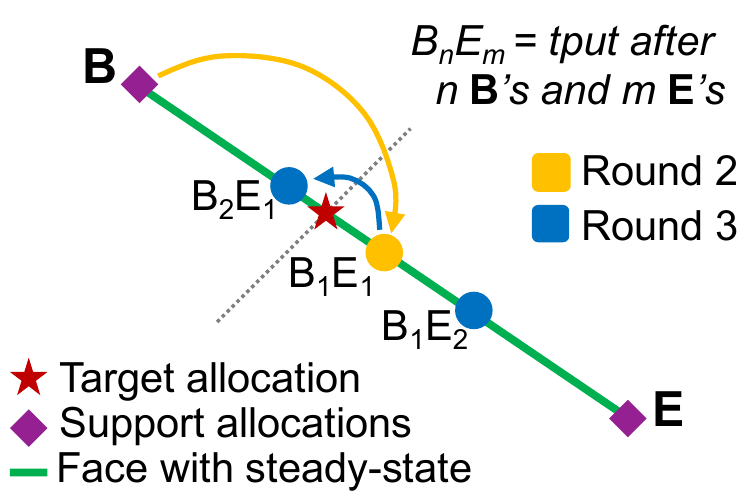}
		\caption{Throughputs in each round.}
		\label{fig:design:rounds}
	\end{subfigure}
  \vspace{-10pt}
  \caption{\name (a) finds the support allocations nearest the target allocation in each round, and (b) converges to the target allocation.}
  \vspace{-10pt}
  \label{fig:design:opt}
\end{figure}

\subsection{Optimality of \name}
\label{sec:design:opt}

\paragraphb{\name is optimal in a round.}
We can prove that \name finds the support allocations nearest the target throughput (in Euclidean distance).
We illustrate this through the example in \fig~\ref{fig:design:ext}, where the corner points of the convex boundary are $\{A, D, B, E, C\}$, and $B$ is closest to the target allocation.
In the previous stage, \name discovered $B$, and it needs to pick one face to continue the search. The shaded yellow regions indicate the extensible regions of the two candidate faces $\overline{AB}$ and $\overline{BC}$. The extensible region of a face describes the space of allocations that can be obtained by searching with the weight vector that defines that face, while maintaining a convex boundary (\S\ref{sec:design:hull}). 
Since \name finds a new allocation on the convex boundary in every stage of the search, no allocation can exist outside these regions; otherwise, the resulting set of discovered allocations would no longer be convex. Thus, the best face for \name to continue the search is indeed $\overline{BC}$, because its extensible region is the only one that may contain a better allocation closer to the $y=x$ line. 
\ifarxiv
\app~\ref{app:opt:boundary} includes
\fi
a formal proof that the optimal support allocation (\ie the allocation closest to the line $y=x$) is unique and that \name finds it. 

\vspace{3pt}
\paragraphb{Optimality over multiple rounds.} 
Under a static task arrival model, we can show that the schedules computed by \name achieve throughputs that are  optimal at the end of every round, i.e., the achieved throughput has the minimum distance possible to the target allocation after  each round. This model assumes the convex boundary remains the same across rounds. One way to realize this is to require (i) the vehicles return to their starting locations at the end of each round, and (ii) all tasks are renewed at the beginning of each round. We make these simplifying assumptions only for ease of analysis; our evaluation in \S\ref{sec:eval} does not use them.

We describe an intuition for this result below.
\ifarxiv
\footnote{See \app~\ref{app:opt:rounds} for a formal proof.}
\fi
Per the static task arrival model, the convex boundary is the same in \emph{each subsequent round}; therefore, \name finds the same support allocations in every round. By taking into account the long-term per-customer rates, $\overline{x}_k(t)$, \name oscillates between these support allocations in each round at the right frequency, such that $\overline{x}_k(t) \; \forall k \in K$ converges to the target allocation over multiple rounds. We illustrate this in \fig~\ref{fig:design:rounds}, which shows the support allocations $B$ and $E$. The face $\overline{BE}$ contains the target allocation, denoted by the star. Because \name oscillates between $B$ and $E$, the allocation $\left(\overline{x}_1(t), \overline{x}_2(t)\right)$ must lie along $\overline{BE}$. \name chooses $B$ in the first round because its throughput is closer than $E$ to the target allocation. In the second round, it chooses $E$, moving the average throughput to $B_1E_1$. In the third round, \name chooses $B$, moving the average throughput to $B_2E_1$. Notice that if it had instead chosen $E$ in the third round, the average throughput would be $B_1E_2$, which is further away from the target throughput. Thus, this myopic choice between $B$ and $E$ results in the closest solution to the target allocation after any number of rounds. Additionally, notice that the length of the jump (\eg from $B$ to $B_1E_1$ and from $B_1E_1$ to $B_2E_1$) decreases in each round; therefore, \name \emph{converges} to the target throughput.

\subsection{Implementation}
\label{sec:design:impl}

We implement the core \name scheduling system in 2,300 lines of Go.\footnote{github.com/mobius-scheduler/mobius} It plugs directly with external VRP solvers implemented in Python or C++~\cite{gurobi, google-vrp}. \name exposes a simple, versatile interface to customers, which we call an \im. An \im consists of a list of desired tasks, where each task includes a geographical location, the time to complete the task once the vehicle has reached the location, and a task deadline (if applicable). In each round, \name gathers and merges \ims from all customers, before computing a schedule. At the end of each round, it informs the customers of the tasks that have been completed, and customers can submit updated \ims. \Ims serve as an abstraction for \name to ingest and aggregate customer requests; however, the merged \im is directly compatible with standard weighted VRP formulations~\cite{pctsp,pdptw} without modification. Thus, \name acts as an interface between customers and vehicles, using a VRP solver as a primitive in its scheduling framework (\fig~\ref{fig:overview}).

\vspace{3pt}
\paragraphb{Bootstrapping VRP solvers.}
Since the VRP is NP-hard~\cite{vrp}, solvers resort to heuristics to optimize \eqn{\ref{eqn:vrp:obj}}. In practice, we find that state-of-the-art solvers do not compute optimal solutions; however, we can aid these solvers with initial schedules that the heuristics can improve upon. We warm-start the VRP solvers with initial schedules generated by the following policies: (i) maximizing throughput, (ii) dedicating vehicles (assuming a sufficient number of vehicles), and (iii) a greedy heuristic that maximizes our utility function (\S\ref{sec:alpha}).
\ifarxiv
\footnote{\app~\ref{app:greedy} describes this heuristic in detail.}
\fi
At the beginning of each round, \name builds a suite of warm start solutions. Then, prior to invoking the VRP solver with some weight vector $\mathbf{w}$, \name chooses the initial schedule from its warm start suite with the highest weighted throughput (\ie objective of \eqn{\ref{eqn:vrp:obj}}). \name also caches the schedules found from all invocations to the VRP solver (\S\ref{sec:design:hull}), to use for warm start throughout the round. \name parallelizes all independent calls to the VRP solver (\eg when computing warm start schedules and when generating $|K|$ schedules to initialize the search along the convex boundary).

    \section{Generalizing to \texorpdfstring{$\mathbf{\alpha}$}{a}-Fairness}
\label{sec:alpha}

The fairness objective we have considered so far aims to provide all customers with the same long-term throughput (maximizing the minimum throughput). However, an operator of a mobility platform may be willing to slightly relax their preference for fairness for a boost in throughput.
To navigate throughput-fairness tradeoffs, we can generalize \name's algorithm (\S\ref{sec:design}) to optimize for a general class of fairness objectives. We use the $\alpha$-parametrized family of utility functions $U_\alpha$, developed originally to characterize fairness in computer networks~\cite{prop-fair}:
\begin{equation}
\label{eqn:utility}
U_\alpha(\mathbf{y}) = \sum_{k \in K} \frac{{y_k}^{1-\alpha}}{1-\alpha},
\end{equation}
where $\mathbf{y} \in \mathbb{R}^{|K|}$ and $y_k$ is the throughput of customer $k$ (either short-term $x_k$ or long-term $\overline{x}_k$). $U_\alpha$ captures a general class of concave utility functions, where $\alpha \in \mathbb{R}_{\geq 0}$ controls the degree of fairness. For instance, when $\alpha = 0$, the utility simplifies to the throughput-maximizing objective defined in Equation~\ref{eqn:vrp:obj} (assuming all customers have the same weight). By contrast, when $\alpha \to \infty$, the objective becomes maximizing the minimum customer's throughput (\ie max-min fairness). $\alpha = 1$\footnote{$U_\alpha$ is not defined at $\alpha = 1$, so we take the limit as $\alpha \rightarrow 1$.} corresponds to proportional fairness, where the sum of log-throughputs of all customers is maximized; this ensures that no individual customer's throughput is completely starved.

\vspace{3pt}
\paragraphb{Generalizing \name's search algorithm.}
When \name generalizes to $\alpha$-fairness, the target allocation is no longer simply the point on the convex boundary that intersects the $y=x$ line. The target allocation is instead the allocation on the convex boundary with the greatest utility $U_\alpha$. When searching the convex boundary in each round, \name determines which candidate face contains the target throughput by using Lagrange Multipliers to find the point \emph{along the face}
\ifarxiv
\footnote{\app~\ref{app:lm} shows how to find the face containing the target throughput.}
\fi
with the greatest utility. Once it finds each support allocation $\mathbf{x}$, \name incorporates the historical throughput $\mathbf{\overline{x}}$ to select the schedule with greatest cumulative utility $U_\alpha\left(\gamma_t \mathbf{x}(t+1) + (1- \gamma_t) \mathbf{\overline{x}}(t)\right)$, where $\gamma_t$ is as defined in \S\ref{sec:design:rounds}.

\vspace{3pt}
\paragraphb{An example.}
\begin{figure}
  \centering
  \includegraphics[scale=0.63]{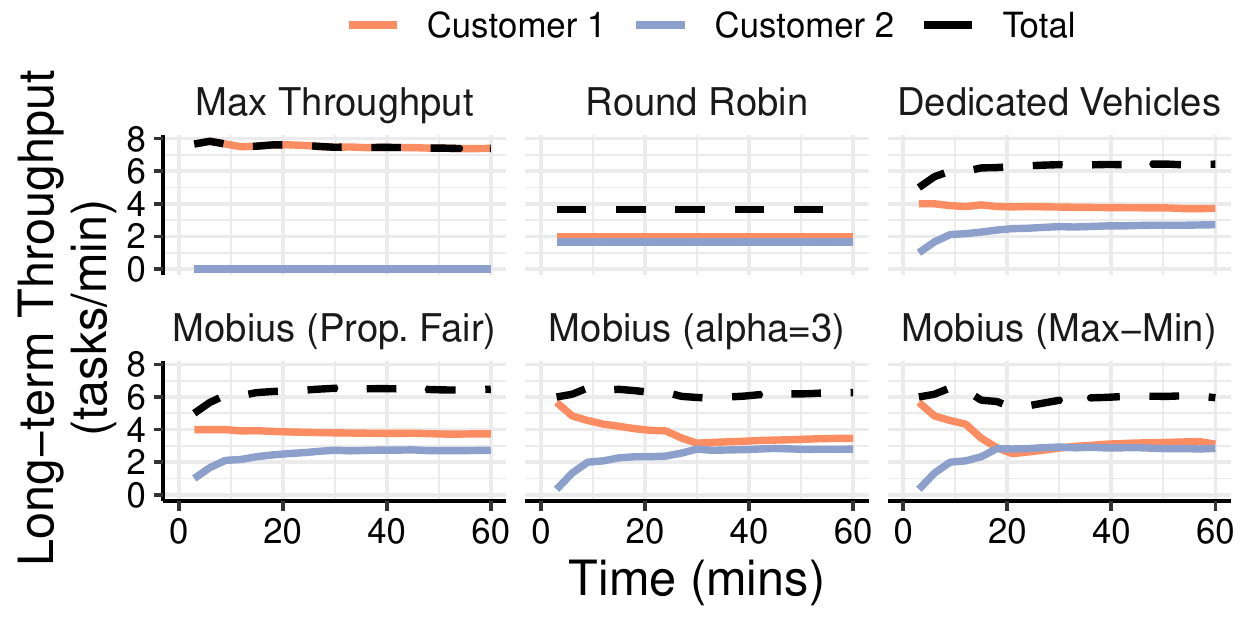}
  \vspace{-10pt}
  \caption{\name can tune its allocation to deliver proportional fairness ($\alpha=1$) and max-min fairness (approximated with $\alpha=100$).}
  \vspace{-10pt}
  \label{fig:alpha:ts}
\end{figure}
\fig~\ref{fig:alpha:ts} shows a time-series chart of long-term customer and platform throughputs for the example described in \S\ref{sec:frontier:dynamics}. By adapting to different schedules on the convex boundary, \name converges to a fair allocation of rates without degrading total throughput. $\alpha$ allows \name to compute \emph{expressive schedules}; for instance, $\alpha=1$ strives to maximize total throughput without starving either customer. Additionally, \name (max-min)\footnote{\name approximates max-min fairness ($\alpha \to \infty$) with $\alpha = 100$.} converges to a fair allocation of long-term throughputs within 20 minutes.

    \section{Real-World Evaluation}
\label{sec:eval}

We evaluate \name using trace-driven emulation (\S\ref{sec:eval:setup}) in two real-world mobility platforms. In \S\ref{sec:eval:lyft}, we apply \name to  Lyft ridesharing in Manhattan and demonstrate that it scales to large online problems. In \S\ref{sec:eval:drone}, we deploy \name on a shared aerial sensing system, involving multiple apps with diverse spatiotemporal preferences. Our evaluation focuses on answering the following questions:
\begin{itemize}
\item How does \name compare to traditional approaches in online scheduling for large-scale mobility problems?
\item How robust is \name in the presence of dynamic spatiotemporal demand from customers?
\item How can we tune \name's timescale of fairness?
\item What other benefits does \name provide to customers, beyond optimizing per-customer throughputs?
\end{itemize}

\subsection{Online Trace-Driven Emulation}
\label{sec:eval:setup}

We implement a trace-driven emulation framework to compare \name against other scheduling schemes, under the same real-world environment. This framework replays timestamped traces of requests submitted by each customer, by streaming tasks to the scheduler as they arrive, and sending task results back to the customer.

\vspace{3pt}
\paragraphb{Capturing environment dynamics and uncertainty.}
To emulate dynamic customer demand, our emulation framework streams tasks according to the timestamps in the trace---so \name has no visibility into future tasks. To emulate uncertainty in customer demand, we cancel tasks that are not scheduled in 10 minutes. Additionally, the case studies in \S\ref{sec:eval:lyft} and \S\ref{sec:eval:drone} consider scenarios where \emph{at least} one customer is backlogged (defined in \S\ref{sec:motivation}). If no customers are backlogged, then the platform can fulfill all tasks within the planning horizon, and the resulting schedule would have maximal throughput and fairness. Thus, the problems are only interesting when at least one customer is backlogged; \name is effective and required only in such situations.

\vspace{3pt}
\paragraphb{Backend VRP solver.} 
We use the Google \ortools package~\cite{google-vrp} as our backend weighted VRP solver (\eqn{\ref{eqn:vrp:obj}}). \ortools is a popular package for solving combinatorial optimization problems, and supports a variety of VRP constraints, including budget, capacity, pickup/dropoff, and time windows. Our case studies involve VRPs with different sets of constraints. We run our experiments on an Amazon EC2 \texttt{c5.9xlarge} instance with 36 CPUs.

\vspace{3pt}
\paragraphb{Baselines.}
In our experiments, we evaluate \name's throughput and fairness against two baseline routing algorithms: (i) a max throughput scheduler, and (ii) dedicated vehicles. The max throughput scheduler simply runs the backend VRP solver on the same input of customer tasks fed into \name for a round. This solution provides a benchmark on the platform capacity, and quantifies the maximum achievable total throughput.
We compute the ``dedicated vehicles'' schedule by first distributing the vehicles evenly among all customers,\footnote{Dedicating vehicles is most suitable when the number of vehicles is a multiple of the number of customers.} and then invoking the max throughput scheduler once for each customer. This solution provides a benchmark schedule that divides vehicle time equally among all customers.
As shown in \S\ref{sec:motivation}, round-robin scheduling achieves very low throughput; hence we omit it from the results in this section. 

To the best of our knowledge, \name is the first algorithm that explicitly optimizes for customer fairness in mobility platforms. We considered evaluating \name by running a scheduler that optimizes throughput and fairness over a longer timescale using a mixed-integer linear program solver (\eg Gurobi~\cite{gurobi} or CPLEX~\cite{cplex}); however, this is not feasible in practice, because (i) customer demands arrive in a streaming fashion, and (ii) these solvers do not scale beyond tens of tasks~\cite{molina2014multi}. Thus, we believe the baselines described above offer reasonable comparisons for \name.

\vspace{3pt}
\paragraphb{Microbenchmarks.}
In addition to the real-world case studies (\S\ref{sec:eval:lyft}-\S\ref{sec:eval:drone}), we also evaluate \name on microbenchmarks created from synthetic customer demand, including scenarios where \name is optimal (under the static task arrival model, \S\ref{sec:design:opt}). We compare \name against max throughput, dedicating vehicles, and round robin, and show, through controlled experiments, that (i) it provides provably good throughput and fairness for a variety of spatial demand patterns, (ii) it scales for different numbers of vehicles, (iii) it controls its timescale of fairness, and (iv) it can tune its fairness parameter $\alpha$. We also report the runtime of \name in various environments. We include these results in
\ifarxiv
\app~\ref{app:runtime} and \app~\ref{app:eval:synthetic}.
\fi

\subsection{Case Study: Lyft Ridesharing in Manhattan}
\label{sec:eval:lyft}

\paragraphb{Setting.}
Motivated by the issue of ``destination discrimination''~\cite{middleton2018discrimination,uber-dd,uber-discrimination} discussed in \S\ref{sec:intro}, we consider a ridesharing service that receives requests from different neighborhoods (customers) in a large metro area. Some neighborhoods are easier to travel to than others, and rider demand out of a neighborhood can vary with the time of day. We show that \name can guarantee a fair quality-of-service (in terms of max-min fair task fulfillment) to all neighborhoods throughout the course of a day, without significantly compromising throughput. We also show that, although it optimizes for an equal allocation of throughputs, \name does not degrade other quality-of-experience metrics, such as rider wait times. We further demonstrate that \name is a scalable \emph{online} platform that generates schedules for a large city-scale problem.

\begin{figure}[t]
	\centering
	\includegraphics[scale=0.65]{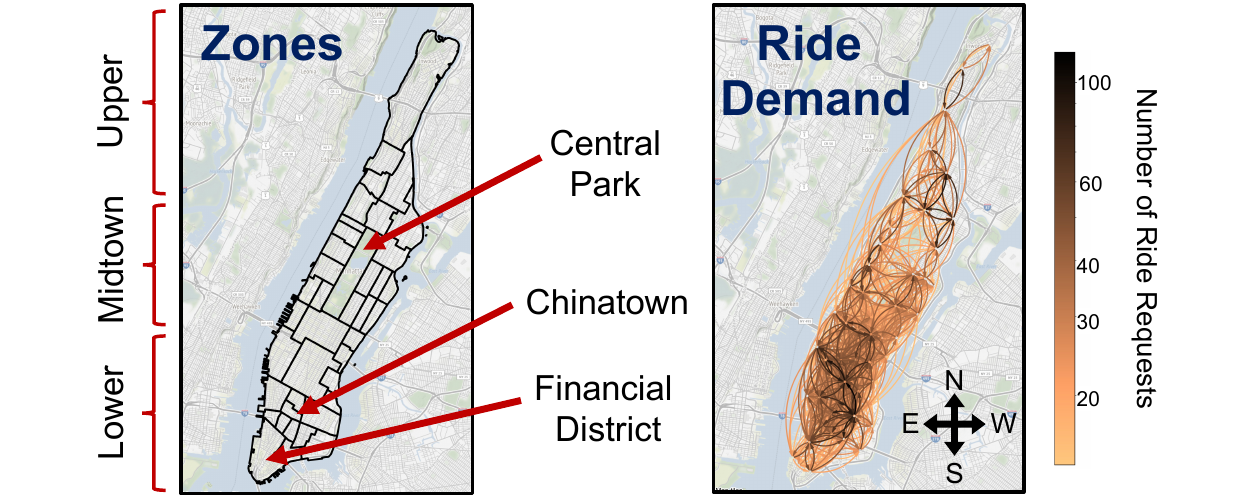}
    \vspace{-10pt}
	\caption{Maps of zones (customers) and demand in Manhattan, indicating skews in both spatial coverage and volume of ride requests.}
    \vspace{-10pt}
	\label{fig:eval:lyft-demand}
\end{figure}

\vspace{3pt}
\paragraphb{Ridesharing demand.}
We use a 13-hour trace of 16,817 timestamped Lyft ride requests, published by the New York City Taxi and Limousine Commission, involving 40 neighborhoods (zones) in Manhattan over the course of a day~\cite{lyft-data}.
Each request consists of a pickup and a dropoff zone, and we seek to provide pickups from all zones equitably. The map in \fig~\ref{fig:eval:lyft-demand} (left) demarcates the customer zones.

\fig~\ref{fig:eval:lyft-demand} (right) illustrates the scale of this scheduling problem. It visualizes traffic on the top 1,000 (out of 3,300) pickup-dropoff pairs; the color of each arrow indicates the volume of ride requests for that pickup-dropoff location. Notice that both the distance of rides and the volume of requests originating from zones vary vastly throughout the island. A significant fraction of requests arrive into and depart from Lower Manhattan. Some zones in Upper Manhattan have as few as 15 unique outbound trajectories, while other zones have hundreds.

Moreover, ridesharing demand varies significantly with the time of day. For instance, a busy zone near Midtown Manhattan sees the load vary from around 200 to 600 requests/hour, and a quiet zone near Central Park experiences a minimum load of 3 requests/hour and peak load of 24 requests/hour. Notice that the dynamic range of demand throughout the 13 hours also varies across zones.

\vspace{3pt}
\paragraphb{Experiment setup.}
This ridesharing problem maps to the capacitated pickup/delivery VRP formulation~\cite{pdptw}. It computes schedules that maximize the total number of completed rides, such that (i) a ride's pickup and dropoff are completed on the same vehicle, and (ii) each vehicle is completing at most one ride request at any point in time. We configure the solver to retrieve real-time traffic-aware travel time estimates from the Google Maps API~\cite{gmaps}, and we constrain \ortools to report a solution within 3 minutes.

We use the trace described above in our emulation framework (\S\ref{sec:eval:setup}). We compute schedules for a fleet of 200 vehicles.\footnote{The number of vehicles does not matter, since we compare \name to the platform capacity (from the max throughput scheduler).} In order to ensure that the schedules are not myopic, we plan our routes with 45-minute horizons; however, to reduce rider wait times, we recompute the schedule every 10 minutes, while ensuring that we honor any requests that we have already committed to in the schedule. We assume that riders cancel requests that are not incorporated into a schedule within 10 minutes of the request time. 

\vspace{3pt}
\paragraphb{Fairness with high vehicle utilization.}
\begin{figure}[t]
	\centering
	\includegraphics[scale=0.65]{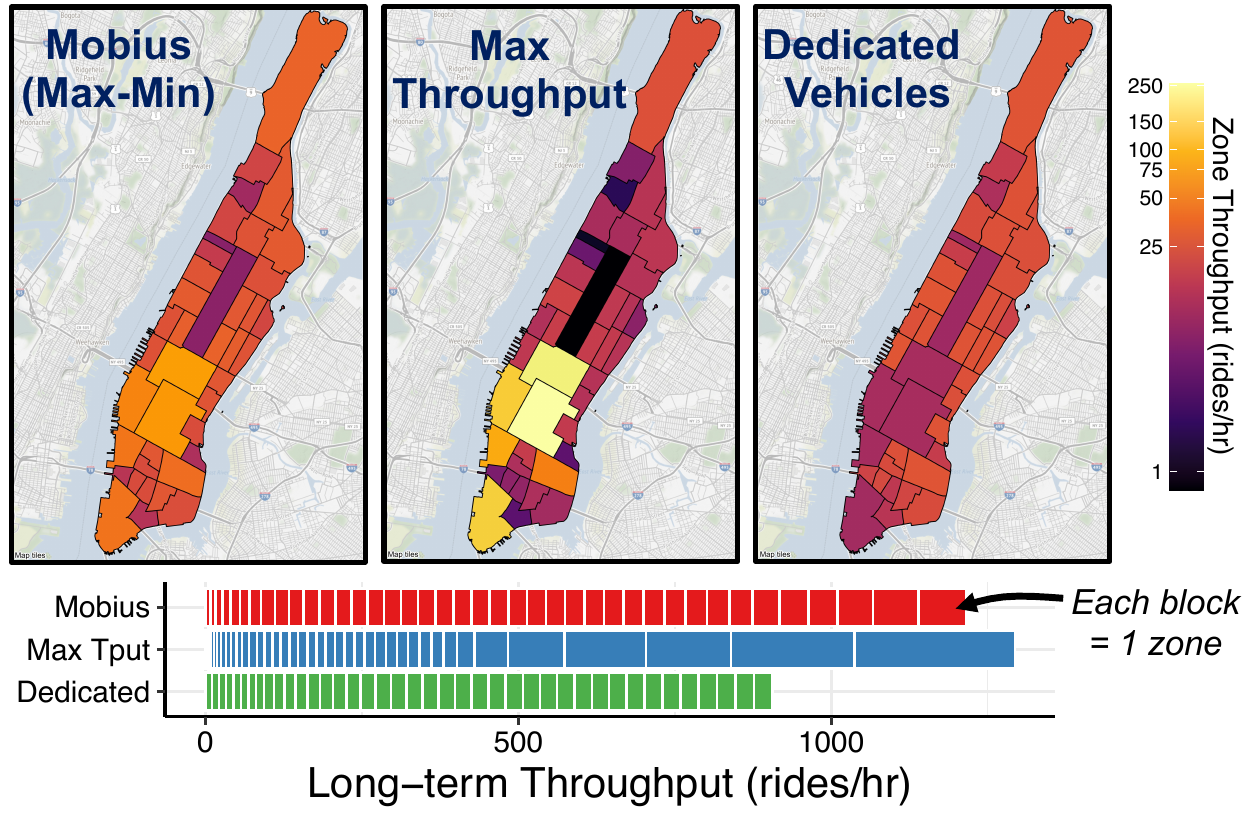}
	\vspace{-10pt}
	\caption{Long-term throughputs for zones in Manhattan after 13 hours. A good scheduler should have a stacked plot with large evenly-sized blocks, and a map with bright (high throughput) and homogeneous (fair) colors across zones.}
    \vspace{-10pt}
	\label{fig:eval:lyft}
\end{figure}
Since \name plans continuously, having several allocations on the convex boundary at its disposal, we expect it to converge to a fair allocation of rates, despite the skew in demand. \fig~\ref{fig:eval:lyft} shows the long-term throughputs achieved for each zone by different scheduling algorithms, after 13 hours. The color of each zone in the map indicates that zone's throughput. Bright colors correspond to high throughput, and a homogeneous mix of colors indicates a fair allocation. Beneath the maps, we also stack the zone throughputs to indicate how each scheduler divides up the total platform throughput across the zones; ideally we would like large, evenly-sized blocks.

The max throughput scheduler divides the platform throughput most unevenly across zones. In particular, we see that while it serves nearly 200 rides/hour out of the Financial District (Lower Manhattan), it virtually starves zones near Central Park. From the demand map (\fig~\ref{fig:eval:lyft-demand}), notice that (i) a majority of rides originate from Lower Manhattan, and (ii) most of these trips are destined for neighboring zones. Thus, the best policy to maximize the total number of trips completed is to stay in Lower Manhattan, which is what the max throughput scheduler does.

The bar chart indicates that dedicating 5 vehicles to each zone results in 40\% lower platform throughput than the max throughput scheduler. This is because a heterogeneous demand across zones cannot be effectively satisfied by an equal division of resources (vehicles).
Nevertheless, \fig~\ref{fig:eval:lyft} shows that this scheduler shares the platform throughput most evenly across zones. The division of per-zone throughputs is not perfectly even, in spite of dedicating an equal number of vehicles, because (i) ride requests from different zones can have different trip lengths, and (ii) some zones have inherently low demand and do not backlog the system, leaving some vehicles idle. 

By contrast, \name strikes the best balance between throughput and fairness. It achieves roughly equal zone throughputs, while compromising only 10\% of the maximum platform throughput. 
Compared to dedicating vehicles, we see, from the map, that \name achieves higher throughput for most zones by identifying an incentive to chain requests from different zones. For example, \name combines two requests from different zones into the same trip, when the dropoff of the first request is close to the pickup of the second request. While this helps improve efficiency, \name also prioritizes pickups from zones with a historically low throughput to ensure fairness across zones.
This ridesharing simulation reveals that  it is possible to achieve a fair allocation of rates in a practical setting \emph{without} significantly degrading platform throughput. 

\vspace{3pt}
\paragraphb{Controlling the timescale of fairness.}
\begin{figure}[t]
	\centering
	\includegraphics[scale=0.65]{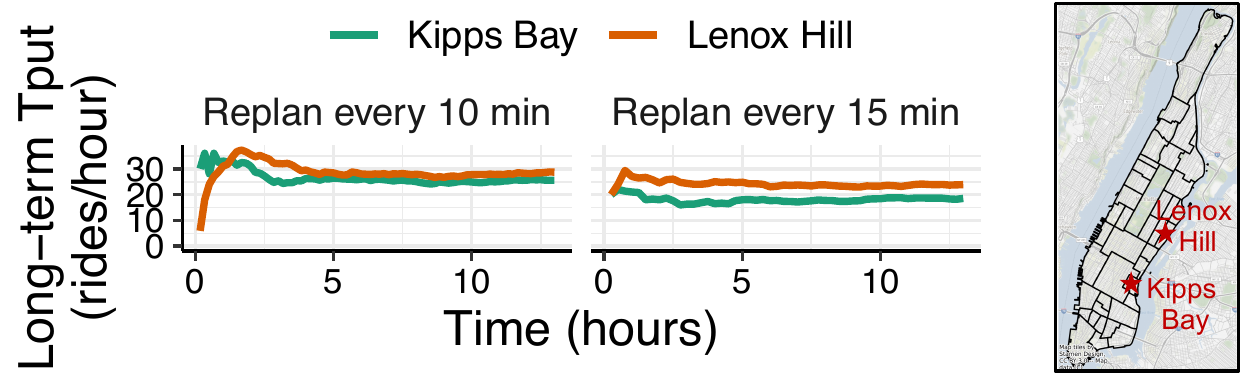}
	\vspace{-10pt}
	\caption{Time series of long-term throughputs for two zones for different replanning horizons. Frequent replanning ensures fairness (equal throughputs) at shorter timescales.}
    \vspace{-10pt}
	\label{fig:eval:lyft-timescale}
\end{figure}
\name's replanning interval controls the timescale over which it is fair. The more often that \name replans, the more up-to-date its record of long-term customer throughputs; \name can then adapt to short-term unfairness quickly by finding a more suitable schedule on the convex boundary. Recall that when replanning frequently, the convex boundary does not change drastically between scheduling intervals (\S\ref{sec:frontier:dynamics}), if the spatial distribution of tasks do not change rapidly with time. So, in practice, we do not expect to deviate far from the ideal target throughput. \fig~\ref{fig:eval:lyft-timescale} shows the long-term throughputs achieved for two zones, for replanning timescales of 10 minutes and 15 minutes. \name  equalizes throughputs better when it replans more frequently.

\vspace{3pt}
\paragraphb{Rider wait times.}
\begin{figure}[t]
	\centering
	\includegraphics[scale=0.65]{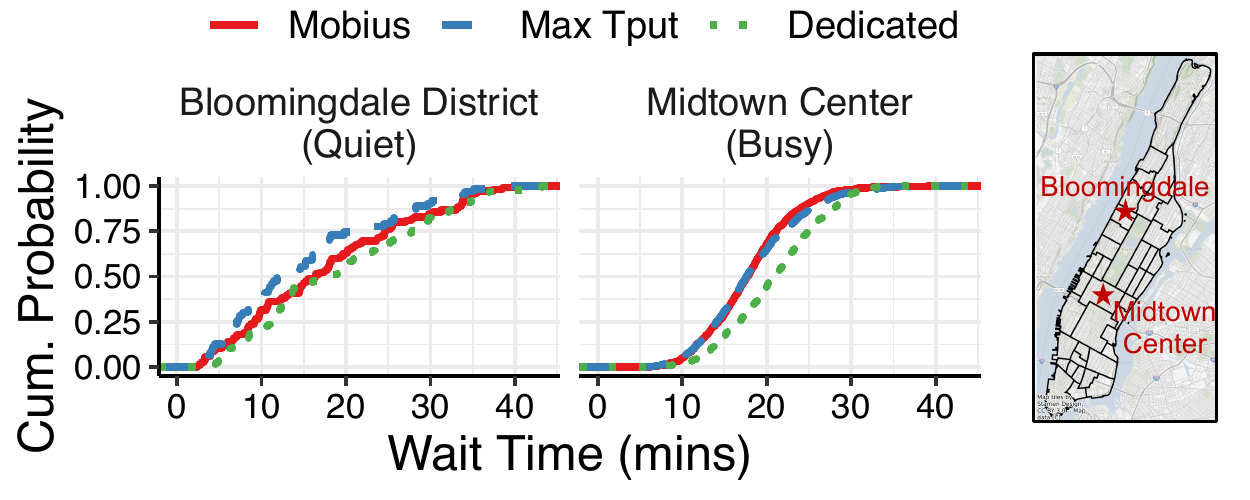}
	\vspace{-10pt}
	\caption{Distributions of rider wait times for two zones. Even though \name compromises some throughput for fairness, it delivers similar wait times as the max throughput scheduler.}
    \vspace{-10pt}
	\label{fig:eval:lyft-wt}
\end{figure}
Platform operators prefer high throughput schedules because they translate directly to high revenue; low throughput would lead to more cancelled rides. 
While \fig~\ref{fig:eval:lyft} demonstrates that \name is fair without degrading throughput, we would like to know if optimizing for fairness impacts rider wait time (\ie the time between requesting a ride and being picked up).

\begin{figure*}
  \centering
  \includegraphics[scale=0.69]{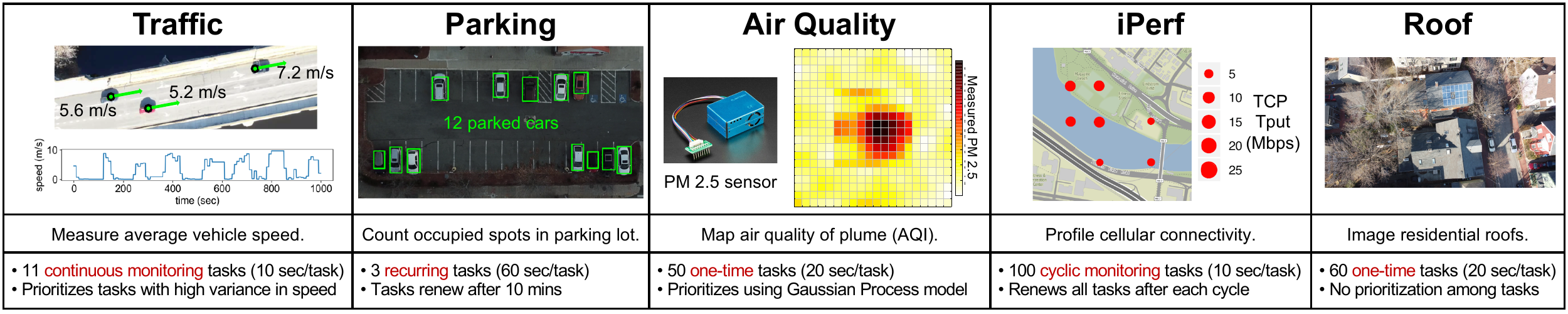}
  \vspace{-10pt}
  \caption{Summary of aerial sensing applications, which span a variety of spatial demand and reactive/continuous sensing preferences. We collected ground truth data for each of these applications using real drones, and created traces to evaluate \name.}
  \vspace{-10pt}
  \label{fig:eval:drone-apps}
\end{figure*}

\fig~\ref{fig:eval:lyft-wt} compares the distributions of rider wait times for rides originating from Bloomingdale District (a quiet neighborhood west of Central Park) and from Midtown Center (a busy district near Times Square). We compute wait times are only for fulfilled tasks.
Notice that in both zones---with two very different demand patterns---the distribution of wait times for \name is comparable to that of the max throughput scheduler. 

We observe that the wait times in the quiet zone are slightly higher for \name (average of 17 minutes, compared with 15 minutes for max throughput). This is because the wait times for \name are computed for significantly more tasks (\name fulfills 66.7\% more ride requests than does max throughput). 
The schedule that dedicates vehicles sees higher wait times than \name, especially when rides originate from a busy zone (\eg Midtown Center), since vehicles would be idle until they return to their assigned zone to pick up a new rider. 

\vspace{3pt}
\paragraphb{Scalability.}
This case study demonstrates that \name is practical at an urban scale. In fact, when scheduling its fleet of taxis, New York City's Yellow Cab restricts its scheduling region to Manhattan and organizes its requests according to approximately 40 taxi zones~\cite{nyc-tlc, vrp-largescale}. In our experiments, the backend VRP solver (\ie max throughput scheduler) computes each 45-minute schedule in 3 minutes (capped by the timeout). We observe that \name takes 5-6 minutes; \name sees a speedup by (i) parallelizing calls to the VRP solver and (ii) warm-starting the VRP solver with initial schedules (\S\ref{sec:design:impl}). These optimizations help \name easily scale to tens of thousands of tasks. We believe we can further improve the speed by leveraging parallelism in the backend VRP solver~\cite{vrp-parallel} (\ortools does not expose a multi-threaded solver). 

\subsection{Case Study: Shared Aerial Sensing Platform}
\label{sec:eval:drone}

\paragraphb{Setting.}
The recent proliferation of commodity drones has generated an increased interest in the development of aerial sensing and data collection applications~\cite{suduwella-mosquitoes, TrackIO, mao-followme, Allison-wildfire, las-surveillance, mersheeva-surveillance}, as well as general-purpose drone orchestration platforms~\cite{beecluster,voltron,petrolo-astro}. An emerging mobility platform is a drones-as-a-service system~\cite{androne, farmbeats, flytos, uav-service, uav-cloud}, where developers submit apps to a platform that deploys these app tasks on a shared fleet of drones. App (customer) semantics in a drone sensing platform can show significant heterogeneity in both space and time.
To ensure a satisfactory \qos for all applications, a scheduler must not only efficiently multiplex tasks from different applications in each flight (typically constrained to 20 minutes due to the battery life~\cite{dji}), but also share task completion throughput equitably across apps. Since apps can be reactive (\ie sensing preferences change as apps receive measured data), \name must additionally provide a sustained rate-of-progress to each app, as opposed to ``bursty'' throughput.

\begin{figure}
  \centering
  \includegraphics[scale=0.5]{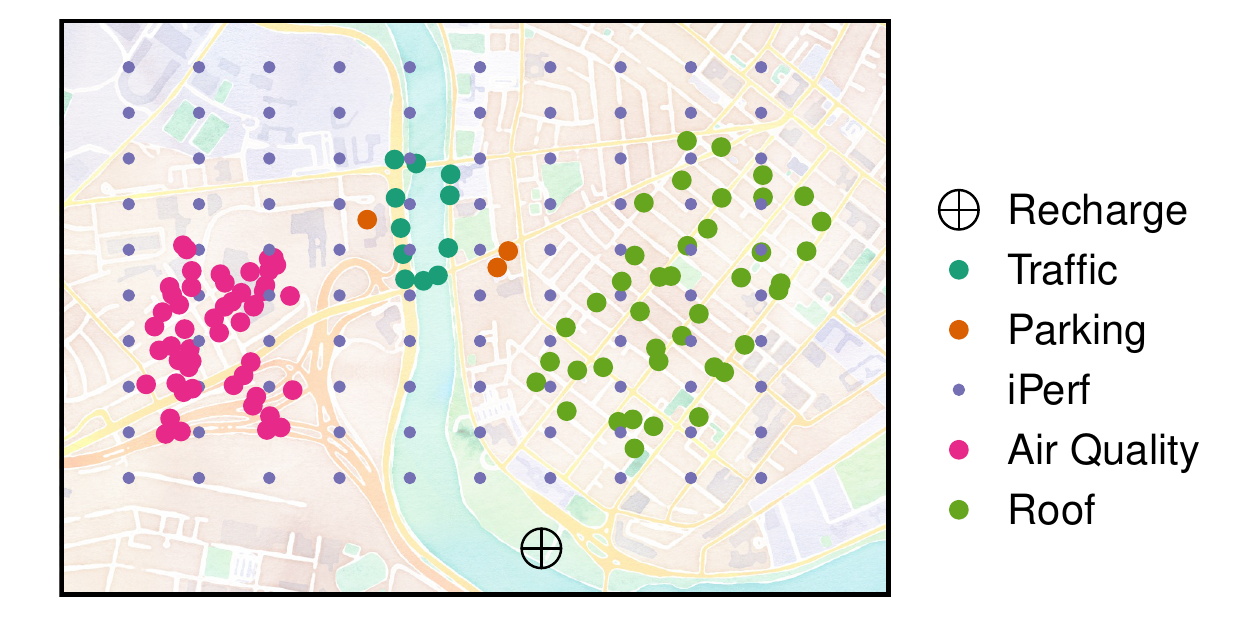}
  \vspace{-10pt}
  \caption{Map of tasks for 5 aerial sensing apps, spanning a 1 square mile area in Cambridge, MA. \name replans every 5 minutes, in order to incorporate new requests. Each drone returns to recharge every 15 minutes.}
  \vspace{-10pt}
  \label{fig:eval:drone-map}
\end{figure}

\vspace{3pt}
\paragraphb{Sensing apps.}
We implement 5 popular urban sensing apps to evaluate \name in this drones-as-a-service context, summarized in \fig~\ref{fig:eval:drone-apps}. \fig~\ref{fig:eval:drone-map} depicts the locations for the sensing tasks submitted by each app. We describe each app below:
\begin{itemize}
\item The \emph{Traffic app} continuously monitors road traffic congestion over 11 contiguous segments of road in an urban area. To measure average vehicle speed, it collects 10-second video clips at each road segment, detects all cars using YOLOv3~\cite{yolo}, and tracks the trajectory~\cite{opencv} of each vehicle. After gathering multiple initial samples at all 11 locations, the app prioritizes the locations with the highest variance in speed, in order to collapse uncertainty in its overall estimates of road congestion. 
\item The \emph{Parking app} counts parked cars at 3 sites, by monitoring each lot for 1 minute; to maintain fresh estimates of counts, this app renews these 3 tasks after 10 minutes. 
\item The \emph{Air Quality app} measures PM2.5 concentration around a plume~\cite{aqi}, submitting a candidate list of 100 one-time sampling locations. This app is also reactive; on receiving a measurement, it updates a Gaussian Process model~\cite{gaussianprocess} and cancels any unfulfilled tasks with high predicted accuracy. 
\item The \emph{iPerf app} builds a map of cellular coverage in the air, by profiling throughput at 100 spatially-dispersed locations. It renews all tasks after each cycle of 100 measurements is complete. 
\item The \emph{Roof app} submits 60 one-time tasks to image roofs over a residential area.
\end{itemize}
Notice that these apps collectively have a variety of spatiotemporal characteristics. For instance, the Traffic app changes its requests with time, based on the uncertainty in speed estimates and the freshness in measurements. By contrast, the Air Quality app changes its requests with space, using a statistical model to collapse uncertainty in a task based on nearby measurements. The iPerf app has no temporal preferences, and instead functions as a ``free-riding'' app that gathers quick measurements over a large area.

\vspace{3pt}
\paragraphb{Ground-truth data collection.} 
To run our drones-as-a-service platform on real-world sensor data, which is critical to the performance of the reactive and continuous monitoring apps, we separately gather 90 minutes of ground-truth data for each app, using real drones. This gives us a trace of timestamped measurement values of each app. We then use our trace-driven emulation framework (\S\ref{sec:eval:setup}) to evaluate different scheduling algorithms. \fig~\ref{fig:eval:drone-apps} shows highlights from our data collection. For example, to collect ground-truth for the Traffic app, we instrument 6 DJI Mavic Pros~\cite{dji} to continuously gather video and track cars over the 11 measurement locations (\fig~\ref{fig:eval:drone-map}) for 90 minutes. Similarly, for the iPerf and Air Quality apps, we program a DJI F450 drone~\cite{djif450} equipped with an LTE dongle and a PM2.5 sensor~\cite{aqi} to gather measurements at their respective measurement locations. We instrument our drone to communicate its location, battery status, and measurement data to a dashboard hosted on an EC2 instance, from which we observe the drone's progress on our laptop.

\vspace{3pt}
\paragraphb{Experiment setup.}
We configure our backend solver to estimate travel time as the Euclidean distance between the sensing tasks plus the sensing time for the destination task. In order to be sufficiently reactive to the Traffic and Air Quality apps, we schedule in 5-minute rounds, and require that the drones return to recharge their batteries every 15 minutes. We run our trace-driven emulation framework with 5 drones. Additionally, we configure the Roof app to join the system after 30 minutes.

\begin{figure}
  \centering
  \includegraphics[scale=0.65]{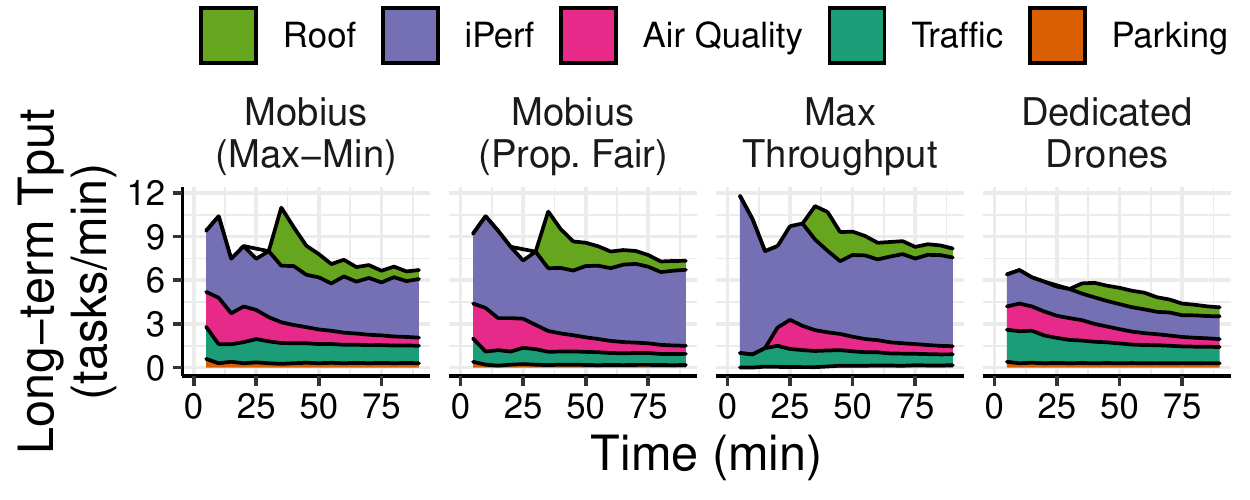}
  \vspace{-10pt}
  \caption{Long-term throughputs achieved over 90 minutes. \name achieves high throughput and best shares it amongst the apps.}
  \vspace{-10pt}
  \label{fig:eval:drone-thp}
\end{figure}

\vspace{3pt}
\paragraphb{High throughput, high fairness.}
To understand how \name divides the platform throughput, we show the long-term throughput for each app over 90 minutes in \fig~\ref{fig:eval:drone-thp}. 
\name (max-min) achieves 55\% more throughput than dedicating drones and only 15\% less throughput than maximizing throughput. 
\name with a proportional fairness objective similarly outperforms max throughput and dedicated vehicles in navigating the throughput-fairness tradeoff.  
Note that the throughputs of the Air Quality and Roof apps decay with time, after their one-time tasks are fulfilled.

Because these apps have variable demand (\eg 100 tasks for iPerf and 3 tasks for Parking), studying throughput is not sufficient. Hence, we plot the tasks completed as a fraction of demand for each app in \fig~\ref{fig:eval:drone-completion}. 
Notice that, under \name, even the most starved app (iPerf) completes nearly 34\% of its tasks; by contrast, max throughput and dedicated drones deliver worst-case task completions of 30\% and 13\%, respectively. Even though dedicating drones guarantees equal drone time for each app, it is extremely unfair toward apps with higher demand or more spatially-distributed tasks. 

\vspace{3pt}
\paragraphb{Impacts of sensing and travel times.}
\fig~\ref{fig:eval:drone-map} would suggest that the Air Quality and Roof tasks are easier to service, since their tasks are more spatially concentrated; however, their tasks take 20 seconds each (\fig~\ref{fig:eval:drone-apps}). The max throughput scheduler understands this tradeoff in terms of maximizing throughput, and thus prioritizes the iPerf app, since its 10-second tasks (\fig~\ref{fig:eval:drone-apps}) are cheap to complete. 
By contrast, \name additionally understands how to navigate this tradeoff in terms of fairness; for instance, it forgoes some iPerf tasks to complete more 20-second AQI measurements.

\begin{figure}
  \centering
  \includegraphics[scale=0.65]{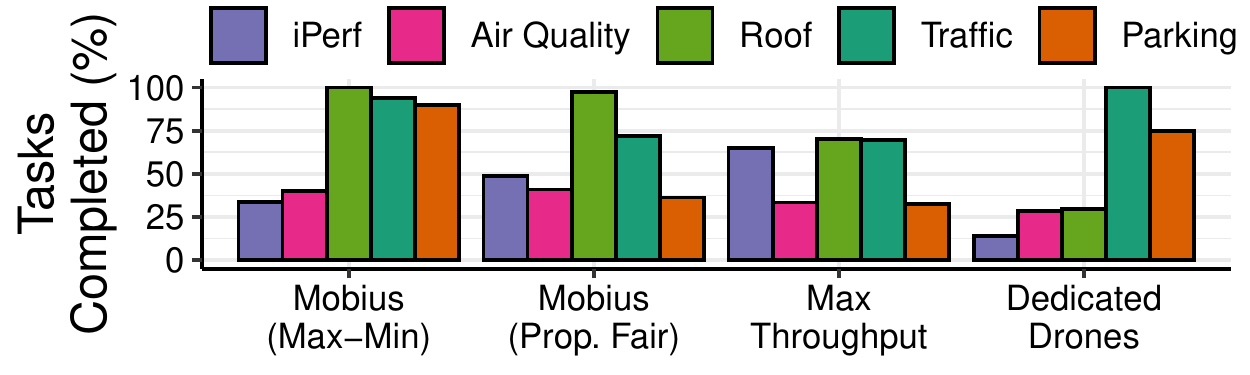}\vspace{-10pt}
  \caption{Percentage of tasks completed per app. \name fulfills nearly all requests for the Traffic and Parking apps, before allocating ``excess'' vehicle time to the more backlogged apps.}
  \vspace{-10pt}
  \label{fig:eval:drone-completion}
\end{figure}

\vspace{3pt}
\paragraphb{Reliable rate-of-progress.}
In enforcing either proportional or max-min fairness, \name does not starve any app, at any instant of time. Indeed, \fig~\ref{fig:eval:drone-thp} indicates that \name delivers a reliable rate-of-progress to the Air Quality app, gradually giving it roughly 3 tasks/min over the first 20 minutes. By contrast, the max throughput scheduler is more ``bursty'', and only services this app after 20 minutes. As a result, we find that, with \name, the root-mean-square error (RMSE) of the Gaussian Process model for the air quality drops more rapidly. 

\vspace{3pt}
\paragraphb{Catering to transient apps.}
\begin{figure}
  \centering
  \includegraphics[scale=0.65]{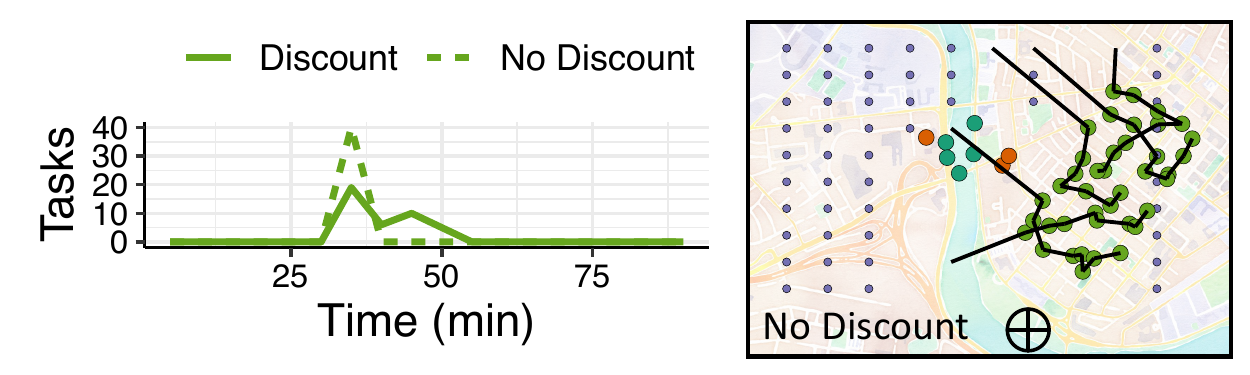}
  \vspace{-10pt}
  \caption{Discounting long-term throughput allows \name to gradually respond to the sudden presence of the transient Roof app, instead of dedicating all drones to it.}
  \vspace{-10pt}
  \label{fig:eval:drone-discount}
\end{figure}
Recall that the Roof app joins the platform after 30 minutes. \fig~\ref{fig:eval:drone-thp} indicates that \name rapidly adapts to this change in demand with a spike in throughput for the Roof app at the cost of lower throughput for the iPerf and Air Quality apps. Notice that this spike in \name's schedule is larger in magnitude than the one in the max throughput schedule. This is because \name realizes that, when the Roof app joins, it has no accumulated throughput, while other apps have amassed higher throughput from living in the system for longer. \fig~\ref{fig:eval:drone-discount} (right) plots the routes for all 5 drones during minutes 30-35; all drones immediately flock to the Roof app. With \name, an operator can choose to respond to the arrival of new apps by discounting throughput accumulated in prior rounds. \fig~\ref{fig:eval:drone-discount} (left) shows how \name can control the Roof app's rate of task fulfillment, with a discount factor of 0.1.

    \section{Related Work}
\label{sec:related}

\paragraphb{Shared mobility and sensing platforms.}
Ridesharing platforms rely on different flavors of the VRP; these systems have typically been interested in maximizing profit (\ie throughput)~\cite{vrp-emptycar, vrp-maxtpt}, minimizing the size of the fleet~\cite{vrp-minfleet}, and planning in an online fashion~\cite{vrp-largescale}. Similarly, there has been a large amount of recent work on drones-as-a-service platforms, which have primarily addressed challenges surrounding data acquisition~\cite{farmbeats}, multi-tenancy and security~\cite{androne}, and programming interfaces~\cite{voltron,beecluster}. All of these systems use a throughput-maximizing algorithm under the hood. \name is motivated by the advent of \emph{customer-centric} mobility platforms in a variety of domains, where guarantees on \qos to customers are paramount to the viability~\cite{uber-dd}
of these services~\cite{middleton2018discrimination}.

\vspace{3pt}
\paragraphb{Vehicle routing problem.}
The VRP has been extensively studied by the Operations Research community~\cite{vrp}. Many variants of the problem have been considered, ranging from the budget-constrained VRP~\cite{pctsp}, capacitated VRP~\cite{vrp-applications}, VRP with time windows~\cite{pdptw}, predictive routing under stochastic demands~\cite{beecluster,vrp-stochastic}, \etc Prior work has extended the VRP to consider multiple objectives, such as minimizing the variance in vehicle travel time or tasks completed by each vehicle~\cite{multiobj-vrp}. These load balancing objectives, however, do not consider customer-level fairness, which is the focus of \name. Moreover, \name abstracts out fairness from the underlying vehicle scheduling problem, making its techniques complementary to the large body of work on the VRP and its variants.

\vspace{3pt}
\paragraphb{Fair resource allocation in computer systems.}
Our approach to formalizing throughput and fairness in mobile task fulfillment is inspired by $\alpha$-fair bandwidth allocation in computer networks~\cite{prop-fair,numfabric}. However, as noted in \S\ref{sec:intro}, mobility platforms introduce new challenges around attributing cost to serve customers, that do not arise when addressing fairness in switch scheduling~\cite{fair-queuing}, congestion control~\cite{kellyfairness}, and multi-resource compute environments~\cite{drf}. \name 
develops a novel set of techniques to address these challenges.

    \section{Conclusion}
\label{sec:conclusion}

We developed \name, a scheduling system that can deliver both high throughput and fairness in shared mobility platforms. \name uses the insight that, when operating over rounds, scheduling on the convex boundary of feasible allocations, as opposed to the Pareto frontier, provably improves on fairness with time. We showed that \name can handle a variety of spatial and temporal demand distributions, and that it consistently outperforms baselines that aim to maximize throughput or achieve fairness at smaller timescales. Additionally, through real-world ridesharing and aerial sensing case studies, we demonstrated that \name is versatile and scalable.

There are several opportunities for extending the capabilities of \name. First, \name assumes that customers are not adversarial. Developing strategyproof mechanisms that incentivize truthful reporting of tasks by customers is an open problem. Second, we design \name to only balance customer throughputs. We believe the optimization techniques we developed (\S\ref{sec:design}) can be extended to support other platform objectives, such as task latency, vehicle revenue, and driver fairness. Finally, incorporating predictive scheduling, where the platform can strategically position vehicles in anticipation of future tasks, is an interesting direction for future work, as it can further increase platform throughput. 

    \begin{acks}
We thank Songtao He, Favyen Bastani, Sam Madden, our shepherd, and the anonymous MobiSys reviewers for their helpful discussions and thoughtful feedback. This research was supported in part by the NSF under Graduate Research Fellowship grant \#2389237. Any opinions, findings, and conclusions or recommendations expressed in this material are those of the authors and do not necessarily reflect the views of the National Science Foundation. The NASA University Leadership Initiative (grant \#80NSSC20M0163) provided funds to assist the authors with their research, but this article solely reflects the opinions and conclusions of its authors and not any NASA entity.
\end{acks}

    \clearpage
    \bibliographystyle{abbrv}
    \bibliography{ref}

    \ifarxiv
        \clearpage
        \appendix

\section{\name Scheduler Design}
\label{app:alg}

\subsection{Searching for \texorpdfstring{$\alpha$}{a}-Fair Allocation}
\label{app:lm}

\S\ref{sec:alpha} explains how \name generalizes its formulation to support a class of $\alpha$-fair objectives. Recall that, in each round, the target allocation is the allocation on the convex boundary with the greatest utility $U_\alpha$; \name finds the support allocations that are closest in utility to the target allocation. In each stage of the search, \name uses Lagrange Multipliers to identify the face containing the allocation that maximizes $U_\alpha$. Specifically, for a face described by the equation $\sum_{k \in K} w_k x_k = c$, \name computes the allocation $\left(x_1^*, \dots, x_{|K|}^*\right)$ with greatest utility, subject to the constraint that it lies on the face. The Lagrangian can be written as:
\[
\mathcal{L}(\mathbf{x}, \lambda) = U_\alpha(\mathbf{x}) - \lambda \left(\sum_{k \in K} w_k x_k - c\right)
\]where $\mathbf{x} \in \mathbb{R}^{|K|}$. To find the utility-maximizing allocation $\mathbf{x^*}$, we set $\nabla\mathcal{L} =  0$, and solve for  $\mathbf{x^*}$:
\[
\nabla\mathcal{L} = 
		\begin{bmatrix}
        \frac{\partial U_\alpha}{\partial x_1^*} - \lambda w_1\\
        \vdots \\
        \frac{\partial U_\alpha}{\partial x_k^*} - \lambda w_k\\
        \vdots \\
        \frac{\partial U_\alpha}{\partial x_{|K|}^*} - \lambda w_{|K|}\\
        \sum_{k \in K} w_k x_k^* - c
        \end{bmatrix}
        = \mathbf{0}
\]
Substituting $\partial U_\alpha/\partial x_k^* = (1/x_k^*)^{\alpha}$, we get:
\begin{equation}
\label{eqn:lm}
x_k^* = (\lambda w_k)^{-1/\alpha}
\end{equation}
To solve for $\lambda$, we substitute \eqn{\ref{eqn:lm}} into the last element of $\nabla \mathcal{L}$:
\begin{equation}
\label{eqn:lambda}
\lambda = \left[\frac{c}{\sum_{k \in K} w_k^{(1-1/\alpha)}}\right]^{-\alpha}
\end{equation}

\name computes $x_k^*$ for the $|K|$ faces that arise from the most recent extension, and identifies the one face that contains $x_k^*$ within the boundaries of its vertices. \name then extends this face in the next stage. For example, in stage 2 in \fig~\ref{fig:design:stages}, we compute $x_k^*$ for faces $\overline{AC}$ and $\overline{CB}$, and continue the search in stage 3 above $\overline{AC}$ because it contains the $x_k^*$ (\ie it intersects the $y=x$ line).

\subsection{Algorithm}
\begin{algorithm}[t]
\caption{\name Scheduler}
\label{alg:mobius}
\begin{algorithmic}[1]
\small

 \Procedure{RunMobius}{$\alpha$}
	\State Initialize history $\mathbf{\overline{x}} \in \mathbb{R}^{|K|}$, with $\mathbf{\overline{x}} = \mathbf{0}$
    \For{each round}
    	\State i $\leftarrow$ InitFace() \Comment{Use basis weights $\mathbf{e_k}$.}
    	\State face $\leftarrow$ \Call{SearchBoundary}{$\alpha$, i}
        \State $\mathbf{x^*} \leftarrow \underset{\mathbf{x} \in \text{face}}{\text{argmax}} \; U_\alpha(\mathbf{x} + \mathbf{\overline{x}})$
        \State Execute schedule with throughput $\mathbf{x^*}$.
        \State $\mathbf{\overline{x}} \leftarrow \mathbf{\overline{x}} + \mathbf{x^*}$
    \EndFor
\EndProcedure

\Function{SearchBoundary}{$\alpha$, face}
	\State $\mathbf{p_{ext}}$ $\leftarrow$ ExtendFace(face) \Comment{Compute $\mathbf{w}$ for face; call VRP.}
	\If{$\mathbf{p_{ext}}$ exists}
    	\For{$\mathbf{p_{face}}$ $\in$ face}
        	\State candidate $\leftarrow$ $\{\mathbf{p_{ext}}\}$ + $\{\mathbf{p} \in$ face | $\mathbf{p} \neq \mathbf{p_{face}}\}$
            \If{\Call{OptInFace}{$\alpha$, candidate}}
            	\State\Return \Call{SearchBoundary}{$\alpha$, candidate}
            \EndIf
        \EndFor
    \Else
    	\State\Return face
    \EndIf
\EndFunction

\Function{OptInFace}{$\alpha$, face}
    \State opt $\leftarrow$ ComputeOpt($\alpha$, face) \Comment{\eqn{\ref{eqn:lm}}}
    \If{opt lies within face}
    	\State\Return true
     \Else 
     	\State\Return false
     \EndIf
\EndFunction

\end{algorithmic}
\end{algorithm}

Algorithm~\ref{alg:mobius} provides pseudocode for \name's scheduling algorithm. When a \maas platform is initialized, \name starts executing the \texttt{RunMobius()} function, first initializing the long-term average throughput $\mathbf{\overline{x}}$. In each round, it computes an initial face using the $|K|$ basis weights $\mathbf{e_k} \; \forall k \in K$, where $\mathbf{e_k}$ is a vector of zeros, with the $k$-th element set to 1. This gives an initial allocation on every axis in the $|K|$-dimensional customer throughput. Then \name runs the \texttt{SearchBoundary()} function to compute the face containing the $|K|$ support allocations around the target throughput. It chooses the allocation $\mathbf{x^*}$ that maximizes the total average throughput.

\section{Optimality of \name}
\label{app:opt}
    
\S\ref{sec:design:opt} provides intuition about the optimality of \name. We show the following results:
\begin{enumerate}[label=\arabic*.]
\item \name gives the optimal solution on the convex boundary in a round.
\item Assuming customer tasks stream in according to a static task arrival model (\S\ref{sec:design:opt}), \name (i) is the optimal allocation on the convex boundary at the end of every around, and (ii) converges to the target allocation with an error that decreases as $\mathcal{O}(1/T)$.
\end{enumerate}
We provide formal mathematical proofs for both results below.

\subsection{\name is Optimal in a Round}
\label{app:opt:boundary}

We denote $H$ as the set of all possible allocations in a round. We show that for every round, \texttt{SearchBoundary()}  (Alg.~\ref{alg:mobius}, line 5) returns a face on the convex boundary of $H$ that maximizes the utility function $U_{\alpha}$. The structure of the proof is as follows:
\begin{enumerate}[label=\arabic*.]
\item Lemma~\ref{app:lem:unique-max} and Corollary~\ref{app:cor:unique-max} identify that the maximizer of $U_{\alpha}$ on $H$ is unique.
\item Lemma~\ref{app:lem:extend-face} shows that any point in the extensible region of a face will not lie above other faces. 
\item In Lemma~\ref{app:lem:max-utility}, we note that extending faces that do not contain the utility maximizing allocation for the stage results in a lower utility. 
\item Finally, we piece together these ideas in Theorem~\ref{app:thm:round-opt} in order to prove the optimality of \texttt{SearchBoundary()} over one round.
\end{enumerate}

\begin{lemma}
For any convex polyhedron $H \in \mathbb{R}^{n^+}$, there is a unique point that maximizes $U_\alpha$ for $\alpha \in (0, \infty]$, and the maximum lies on the boundary of $H$.
\label{app:lem:unique-max}
\end{lemma}
\begin{proof}
$U_{\alpha}$ is strictly concave for any finite $\alpha$. When maximizing a concave function over a convex set, the optimum point is unique, the local maximum is the global maximum, and the optimal point is on the boundary of $H$~\cite[Theorem 8.3]{opt-intro}.
\end{proof}

\begin{corr}
There exists exactly one candidate face at any stage of the convex boundary (among all possible faces in Alg. \ref{alg:mobius}, line 14) for which \texttt{OptInFace()} is true.
\label{app:cor:unique-max}
\end{corr}
\begin{proof}
It follows from Lemma 1 that the optimal $U_\alpha$ over the current convex boundary is unique. This means that exactly one face must have the optimal within its face. 
\end{proof}

\begin{definition}
A point $\mathbf{p}$ is said to be above (or below) a face described as
$\sum_{k\in K} w_k x_k = c$ if $\sum_{k\in K} w_k p_k > c$ (or $\sum_{k\in K} w_k p_k < c$). 
\end{definition}

\begin{lemma}
Let $f \in F$ be a face among all candidate faces $F$ during a given stage. Any allocation resulting from an initial call to \texttt{SearchBoundary($f$)} will never lie above any face in $F\setminus \{f\}$.
\label{app:lem:extend-face}
\end{lemma}
\begin{proof}

We prove this by contradiction. Suppose $\mathbf{p}$ was above two faces $f$ and $\tilde{f}$. Then there will exist at least one support allocation $\mathbf{x}$ of face $\tilde{f}$ which could not have been found from a previous call to \texttt{ExtendFace()}, and $\mathbf{p}$ would have been found instead of $\mathbf{x}$. For example, in \fig~\ref{fig:design:search}, any point above the extensible regions of $\overline{AB}$ and $\overline{BC}$ would contradict the existence of $B$. 

Thus, $\mathbf{p}$ cannot lie above more than one face. This argument is true for any subsequent calls to \texttt{SearchBoundary()}, and any subsequent allocation obtained from extending a face has to be below all other faces.
\end{proof}

\begin{lemma}
Let $\mathbf{p}$, contained in face $f_p$, be the maximizer of $U_\alpha$. The utility of any allocation in the extensible region of a face $f \neq f_p$ will be lower than the utility at allocation $\mathbf{p}$.
\label{app:lem:max-utility}
\end{lemma}
\begin{proof}
Since $U_\alpha$ is concave, and is maximized at point $\mathbf{p}$ on $f_p$, the value of $U_\alpha$ will only increase if evaluated at a point above $f_p$. From Lemma \ref{app:lem:extend-face}, we know that every extension of a face other than $f_p$ will be below $f_p$, and thus have a lower utility than the current value evaluated at $\mathbf{p}$.
\end{proof}

\begin{theorem}
In each round, \texttt{SearchBoundary()} returns the face on the convex boundary that contains the allocation that maximize $U_{\alpha}$.
\label{app:thm:round-opt}
\end{theorem}

\begin{proof}
In every round, \texttt{SearchBoundary()} iteratively extends the face that contains the optimal $U_\alpha$ over all faces. Lemma~\ref{app:lem:max-utility} guarantees that any subsequent exploration of a face required by \texttt{SearchBoundary()} will only increase $U_\alpha$, while pruning out the search spaces which cannot improve the solution. When maximizing a concave objective function over a convex set, the local optimal is the global optimal. Therefore, the solution returned when \texttt{SearchBoundary()} terminates is the maximum of $U_\alpha$ over $H$.
\end{proof}

\subsection{\name Converges to the Target Throughput}
\label{app:opt:rounds}
In our problem setting, customer tasks are only presented to \name for the current round, and no knowledge of future task locations is assumed.\footnote{A greedy approach (discussed in \app~\ref{app:greedy}) would be a regret-free online algorithm for this planning problem.} In this section, we show an interesting result: under the static task arrival model (\S\ref{sec:design:opt}), \name, although myopic, results in allocations that are globally optimal. In other words, asymptotically, \name achieves the {\em same average throughput allocation} that a utility-maximizing oracle, which jointly planned over multiple rounds, would achieve.

\begin{definition}
A task distribution is said to be static if the set of throughput allocations (denoted as $H$) is the same for all rounds.
\label{app:def:stationary}
\end{definition}

\begin{definition}
The maximum of $U_\alpha$ over the convex boundary of $H$ is defined as the optimal long term throughput, and is denoted as $\mathbf{x}^{*}$. 
\end{definition}

The set of all feasible average throughput allocations after $t$ rounds is denoted as $F_t$. We prove the following: (i) in every round, \name chooses the solution which maximizes $U_\alpha (\mathbf{\overline{x}}(t))$ on the convex boundary of $F_t$, and (ii) $\mathbf{\overline{x}}(t)$ converges to $\mathbf{x^{*}}$ with an error that decreases as $1/t$.
For brevity, we consider the case with two customers ($|K| = 2$). However, the intuition and results generalize for any number of customers. The proof outline is as follows
\begin{enumerate}[label=\arabic*.]
\item Lemma \ref{app:lem:sameCH} and \ref{app:lem:denseCH} characterize the evolution of the convex boundary
\item Lemma \ref{app:lem:opt-in-CH} proves that $\mathbf{x^{*}}$ lies on this boundary.
\item Lemma \ref{app:lem:Mobius-is-optimal} shows that \name is optimal at the end of any finite round $t$.
\item Finally, we prove asymptotic optimality and describe the rate of convergence in Theorem \ref{app:thm:round-converge}.
\end{enumerate}

\begin{lemma}
For any round $t$, the convex boundary of $F_t$ remains constant. 
\label{app:lem:sameCH}
\end{lemma}
\begin{proof}
The allocations in $F_t$ are obtained by averaging the throughput obtained over $t$ rounds of $H$ (see Fig. \ref{fig:frontier:bf-rate}). Since $H$ is convex, any average of allocations in $H$ will remain the same boundary. Thus, the convex boundary of $H$ and $F_t$ are the same.
\end{proof}

\begin{lemma}
At round $t$, each face of the convex boundary contains $t-1$ equidistant allocations.
\label{app:lem:denseCH}
\end{lemma}
\begin{proof}
Consider one face $f$ in the convex boundary of $H$. From Lemma \ref{app:lem:sameCH}, we know that the convex boundary of the average throughput at any subsequent round will remain the same (see Fig. \ref{fig:frontier:bf-rate}). However, with subsequent rounds, linear combinations of two corner points that constitute a face will create new allocations that lie along the same face. For example, in \fig~\ref{fig:design:rounds}, the face $\overline{BE}$ gets ``denser'' with more allocations, as the number of rounds increases. Remember that in every round, \name can only choose between throughput allocations $B$ or $E$ to modify the average. In particular, by round $t$, $n$ allocations of $B$ and $m$ allocations of $E$ result in an allocation $B_n E_m = \frac{n}{t} B + \frac{m}{t} E$, where $n$ and $m$ are integers. Now, since there are $t-1$ possible combinations of $n$ and $m$ that sum up to $t$, we get $t-1$ equidistant allocations on the face.
\end{proof}

\begin{lemma}
$\mathbf{x}^{*}$ lies on the face in the convex boundary of $H$.
\label{app:lem:opt-in-CH}
\end{lemma}
\begin{proof}
The maximizer of $U_\alpha$ in the long term searches over the convex boundary of $H$. Since $U_\alpha$ is a concave function, and the search space is a convex set, $\mathbf{x}^{*}$ must lie on the face of the convex boundary (Lemma~\ref{app:lem:unique-max}).
\end{proof}

\begin{figure}
    \centering
    \includegraphics[width = 0.9\linewidth]{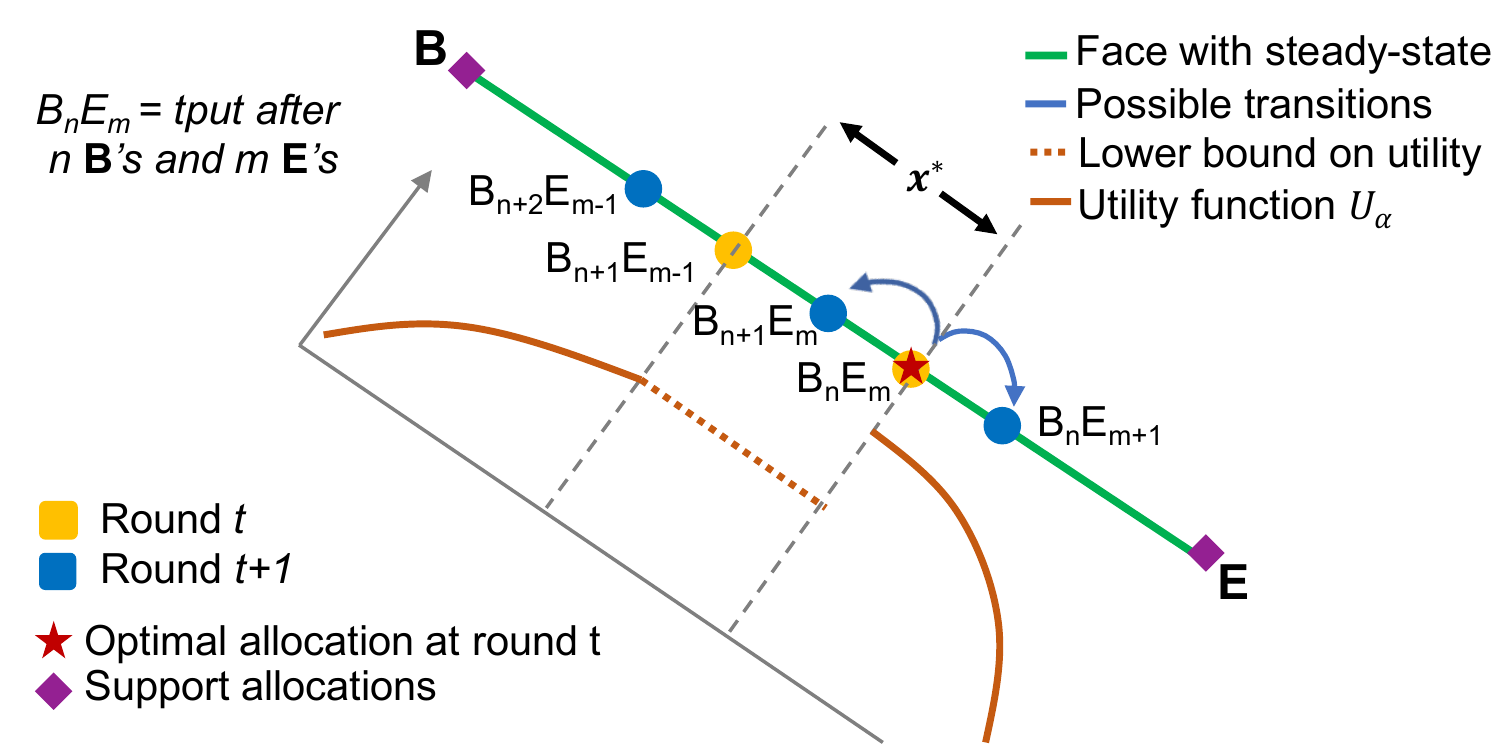}
    \caption{Proof setup for Lemma \ref{app:lem:Mobius-is-optimal}. The face $\overline{BE}$ is the same as in Fig. \ref{fig:design:opt}}
    \label{fig:app:optimality-proof}
\end{figure}

\begin{lemma}
For any round $t$, \name finds the utility-maximal allocation on the convex boundary of $F_t$.
\label{app:lem:Mobius-is-optimal}
\end{lemma}
\begin{proof}
We prove this by induction on the number of rounds. Theorem~\ref{app:thm:round-opt} proves the base case, where $t=1$. Suppose \name finds the highest utility solution on the convex boundary of $F_t$. We want to show that the schedule that \name computes (Algorithm~\ref{alg:mobius}) has the highest utility on the convex boundary of $F_{t+1}$. We use the example in \fig~\ref{fig:design:rounds} to build this argument. Suppose $B_{n}E_{m}$ has the highest utility after $t = n + m$ rounds. Without loss of generality, we assume that the optimal allocation $\mathbf{x}^{*}$ lies between $B_{n}E_{m}$ and $B_{n+1}E_{m-1}$ (rather than $B_{n}E_{m}$ and $B_{n-1}E_{m+1}$).

We now argue that \name chooses $B$ or $E$ appropriately at round $t+1$ to ensure that the average throughput is the optimal at the end of round $t+1$. Figure \ref{fig:app:optimality-proof} illustrates the setup for this proof.

Recall that $U_\alpha$ is a concave function. Thus, the utility function evaluated over the face $\overline{BE}$ is also concave. This implies that (i) the utility at $B_{n}E_{m}$ is higher than any allocation in $(B_{n}E_{m}, E]$, and (ii) the utility at $B_{n+1}E_{m-1}$ is higher than any allocation in $[B, B_{n+1}E_{m-1})$. Also, since $B_{n}E_{m}$ was the optimal solution at round $t$ (by the inductive hypothesis), the utility at $B_{n}E_{m}$ is higher than the utility at $B_{n+1}E_{m-1}$. Thus the utility for any allocation in $[ B_{n+1}E_{m-1} , B_{n}E_{m} ]$ is lower bounded by the utility at $B_{n+1}E_{m-1}$.

The above relations prove that in round $t+1$, $B_{n+2}E_{m-1}$ can never have the highest utility. Thus the optimal throughput on the convex boundary of $F_{t+1}$ is either $B_{n+1}E_{m}$ or $B_{n}E_{m+1}$. This results in two cases:
\begin{itemize}
\item $B_{n+1}E_{m}$ is the optimal, in which case \name would choose allocation $B$ for round $t+1$ to reach the optimal.
\item $B_{n}E_{m+1}$ is the optimal, in which case \name would choose allocation $E$ for round $t+1$ to reach the optimal.
\end{itemize}
Note that by eliminating $B_{n+2}E_{m-1}$ as an optimal allocation, the two remaining candidates can be reached by appropriately choosing $B$ or $E$ in round $t+1$. This concludes the induction argument and proves the lemma.
\end{proof}

Note that Lemma~\ref{app:lem:Mobius-is-optimal} proves that at the end of round $t$, \name not only reaches the best allocation at the end of round $t$, but it also achieves the best allocation for each preceding round before $t$. 

\begin{theorem}
\name (i.e. Alg. \ref{alg:mobius}) converges to $\mathbf{x^*}$ such that the distance between $\overline{\mathbf{x}}$ and $\mathbf{x^*}$ decreases as $\mathcal{O}(1/t)$, where $t$ is the number of rounds.
\label{app:thm:round-converge}
\end{theorem}
\begin{proof}
From Lemma \ref{app:lem:denseCH}, we know that at any round $t$ there are $t-1$ feasible cumulative allocations (excluding the extreme points) on a face, and that these allocations split the face into equally-spaced segments of length $\propto \frac{1}{t}$. Thus, the best allocation in $F_t$ converges to the optimal $\mathbf{x}^{*}$ with an error that is bounded by $\mathcal{O}(1/t)$. Since Lemma~\ref{app:lem:Mobius-is-optimal} establishes that \name chooses the optimal allocation in $F_t$, the result follows.
\end{proof}

\section{Greedy Heuristic to Maximize \texorpdfstring{$U_\alpha$}{U\_a}}
\label{app:greedy}

\S\ref{sec:design:impl} describes how \name builds a suite of warm start schedules to assist the VRP solver in maximizing a weighted sum of customer throughputs. Since \name is guided by a utility function $U_\alpha$ (\S\ref{sec:alpha}), we implement a heuristic that computes an $\alpha$-fair schedule by performing a greedy maximization of $U_\alpha$. Note that this algorithm is not guaranteed to result in a schedule on the convex boundary; we instead use it as an initial schedule to warm start the VRP solver (\S\ref{sec:design:impl}).

The greedy heuristic uses the same formulation as the VRP (\S\ref{sec:motivation}). It computes routes for each vehicle $v \in V$ subject to the budget constraints. \name constructs a schedule iteratively; in each iteration, it executes two steps:
\begin{enumerate}
\item It constructs an $\alpha$-fair path \emph{for each vehicle} that meets the budget constraints by trying to maximize $U_\alpha$. 
\item Then, it invokes a VRP solver with the tasks selected by the paths in (1), to build a high throughput schedule to complete the fair allocation of tasks.
\end{enumerate}
At the end of each iteration, it takes the VRP schedule generated by (2) and tries to squeeze more tasks into the path, preserving fairness. It terminates when no new task can be added according to the greedy optimization in step (1). It then runs the VRP one final time, with a very high weight on the final set of $\alpha$-fair tasks and lower weight on all other customer tasks, so that it can pack the schedule with more tasks to achieve a schedule with high total throughput. 

Intuitively, the iterations over steps (1) and (2) create the $\alpha$-fair schedule with the highest possible throughput, according to the greedy approximation. Then, with the final packing step, we try to boost the throughput of the schedule by fulfilling any additional tasks, without compromising the $\alpha$-fair allocation we have already committed to.

Before starting to construct a path iteratively, we internally maintain the total number of tasks $h_k$ currently fulfilled by the path, for each customer. To do a greedy maximization of $U_\alpha$ in each iteration of constructing the path, we sort all customer tasks according to the \emph{return-on-investment} $R$ for completing each task. Recall that $T_k$ is the set of tasks requested by customer $k$ and $B$ is the time budget for a round (\S\ref{sec:motivation}). The new throughput for customer $k$ by fulfilling a task $l \in T_k$ is $x_k = (h_k + 1)/B$. We compute a task $l \in T_k$ as:
\begin{equation}
R(l) = \frac{U_\alpha(\mathbf{x}) - U_\alpha(\mathbf{h})}{c(m,l)}
\end{equation}
where $m$ is the last task in the path and $c(\cdot)$ is the cost to travel from $m$ to $l$. Then, to select the next task $l$ to add to the path we simply find the task with the greatest return-on-investment:
\begin{equation}
\underset{l \in T_k \; \forall k \in K}{\text{argmax}} \;\;\;\; R(l)
\end{equation}
\section{Runtime of \name}
\label{app:runtime}
In each round, \name uses an efficient algorithm to find support allocations near the target allocation, without having to compute all corner points of the convex hull in each round. This allows \name to invoke the VRP solver \emph{sparingly} in order to find an allocation of rates that steers the long-term rates toward the target. We report some highlights from profiling \name in Table~\ref{tab:app:prof}. These results suggest that the computational overhead for deploying \name would be negligible for many mobility applications.
\begin{table}[ht]
    \footnotesize
	\centering
	\begin{tabular}{c|c|c|c|c}
		\# Cust. & \# of Tasks & \# of Vehicles & Round Duration (min) & Runtime (s) \\
		\hline
		\hline
		2 & 100 & 2 & 10 & 15 \\
        3 & 150 & 3 & 15 & 20 \\
        4 & 200 & 4 & 15 & 35 \\
        6 & 567 & 6 & 90 & 51 \\
        6 & 999 & 6 & 90 & 59 \\
        6 & 567 & 24 & 90 & 88 \\
        6 & 999 & 24 & 90 & 105 \\
	\end{tabular}
	\caption{Performance of \name on different input sizes.}
	\label{tab:app:prof}
\end{table}

\section{Microbenchmarks}
\label{app:eval:synthetic}

\begin{figure*}[t]
	\centering
    \begin{subfigure}[t]{0.15\textwidth}
		\centering
		\includegraphics[scale=0.55]{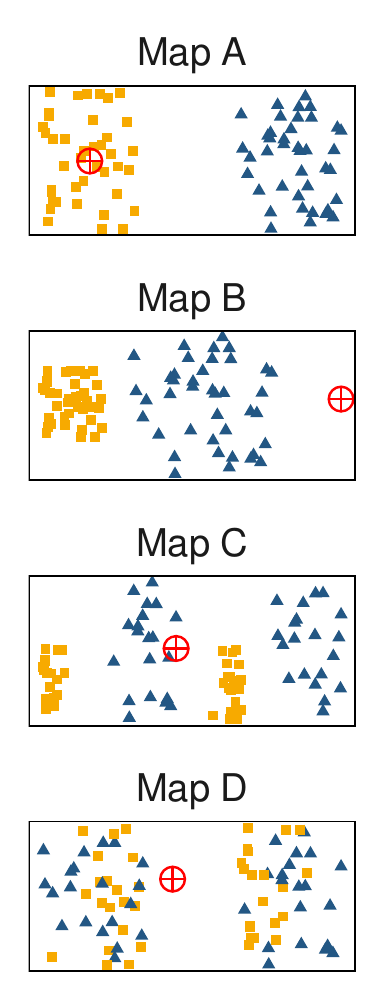}
		\caption{Map of tasks.}
		\label{fig:eval:synthetic-cust2-map}
	\end{subfigure}
    \begin{subfigure}[t]{0.4\textwidth}
		\centering
		\includegraphics[scale=0.55]{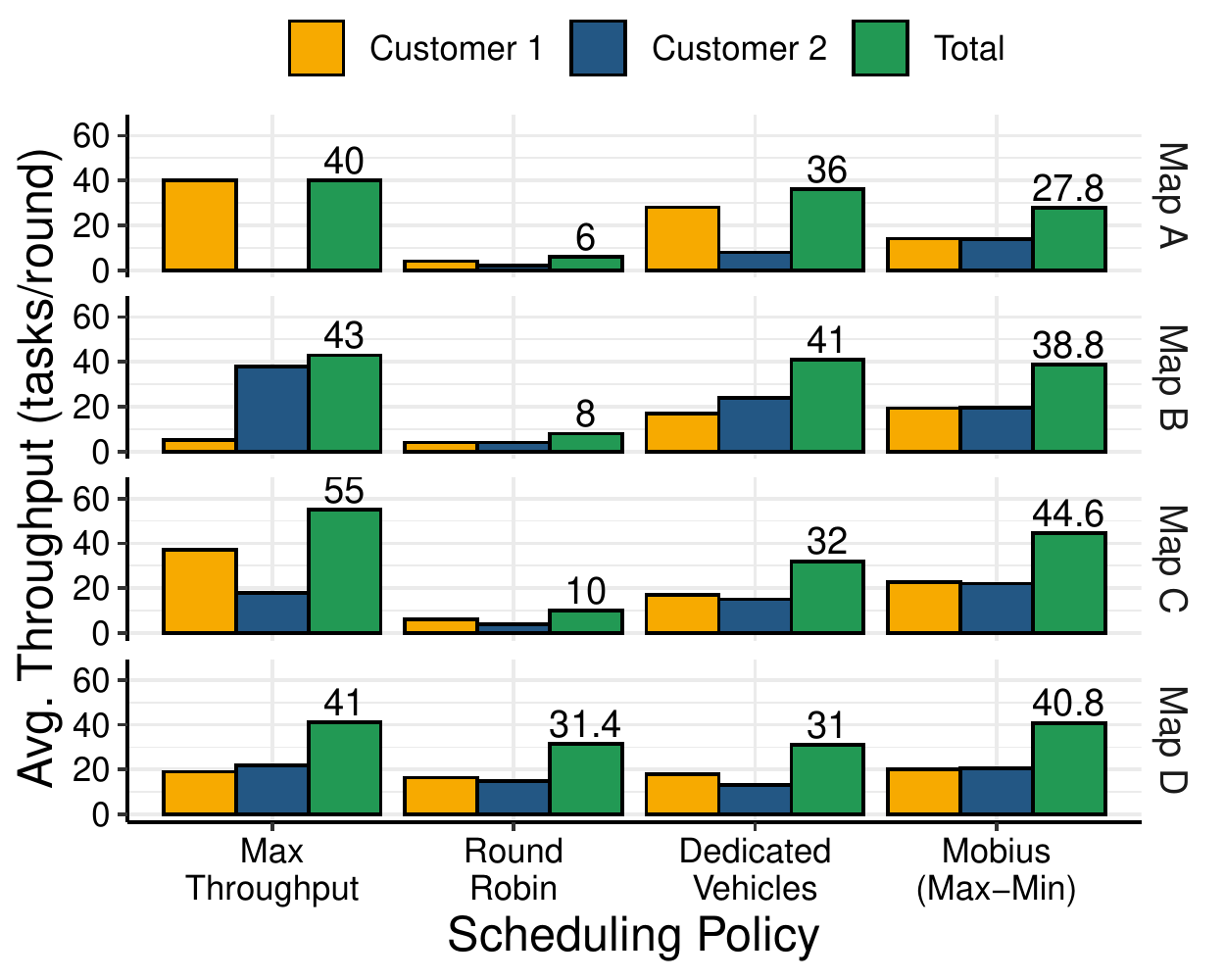}
		\caption{\name vs. other schemes.}
		\label{fig:eval:synthetic-cust2-sum}
	\end{subfigure}
      \begin{subfigure}[t]{0.4\textwidth}
      \centering
      \includegraphics[scale=0.55]{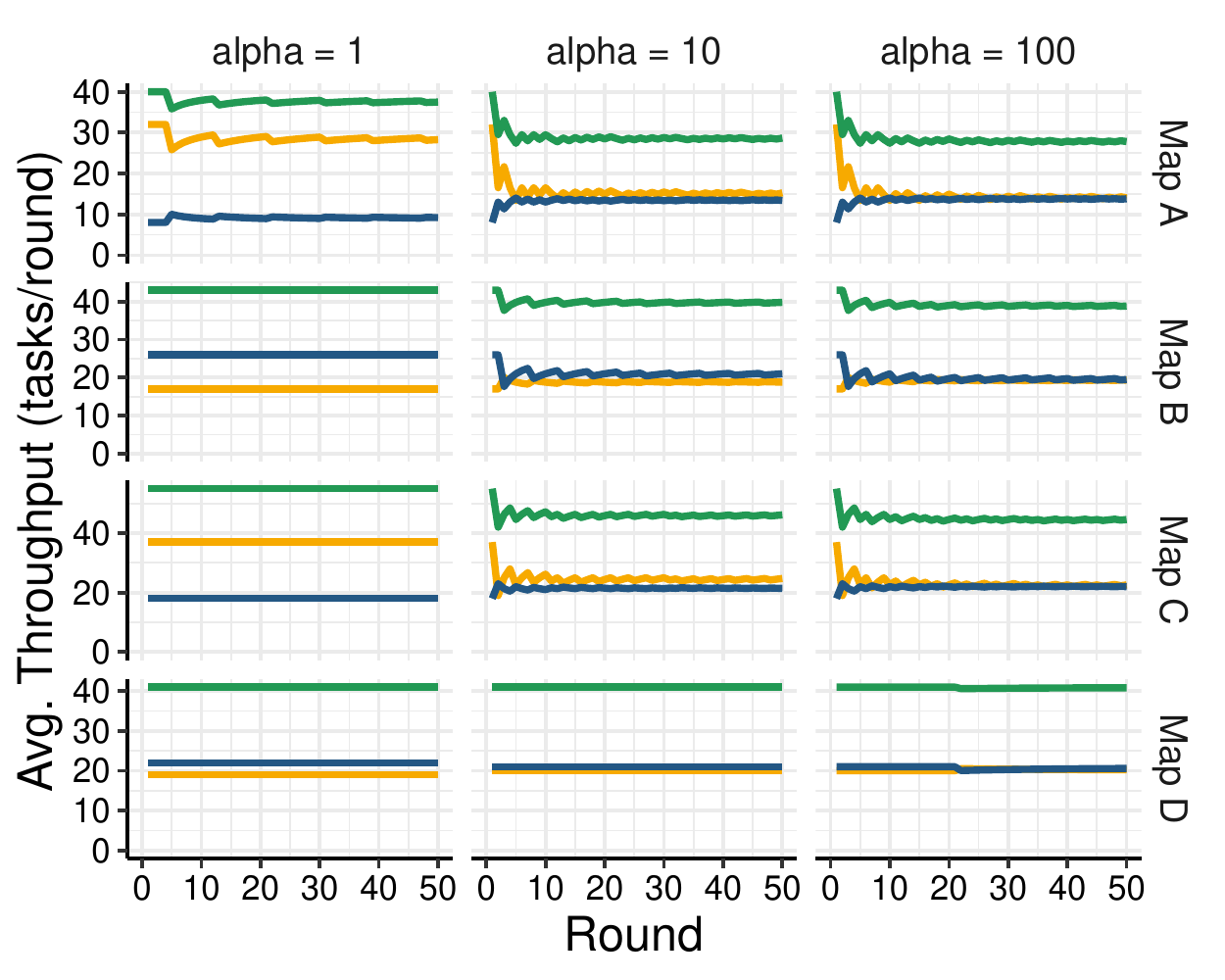}
      \caption{$\alpha$-fair throughputs achieved by \name.}
      \label{fig:eval:synthetic-cust2-thp}
      \end{subfigure}
	\caption{Comparing customer throughputs and platform throughput achieved by \name and other schemes. Customer tasks stream in according to a static task arrival model. \name consistently outperforms other schemes by striking a balance between throughput and fairness.}
    \vspace{-10pt}
	\label{fig:eval:synthetic-cust2}
\end{figure*}

We evaluate \name on several microbenchmarks involving synthetic customer traces. We vary (1) the spatial distribution of customer tasks, (2) the timescale at which tasks are requested, and (3) the degree $\alpha$ to which a schedule is fair. 
Customers submit at most 40 tasks, each taking 10 seconds to fulfill, in any round. Between rounds, they renew any fulfilled tasks at the same location. 
 The travel time between any two nodes is based on their Euclidean distance, assuming a constant travel speed of 10 m/s.

\subsection{Robustness to Spatial Demand}
We evaluate \name's ability to deliver a fair allocation of customer throughputs, in the presence of highly diverse spatial demand. We construct 4 very different maps (\fig~\ref{fig:eval:synthetic-cust2-map}), with 2 vehicles starting at $\oplus$. For this experiment, we assume that the customer tasks arrive from a static task arrival model (\S\ref{sec:design:opt}): the vehicles make a round-trip in each round, and customers renew any fulfilled tasks at the start of every round.
\fig~\ref{fig:eval:synthetic-cust2-sum} shows the average per-customer and total throughputs achieved by different schemes after 50 rounds. 

For all maps, \name (with max-min fairness) indeed provides a fair allocation of average customer throughputs. The other baseline schedules exhibit variable performance depending on the task distribution. For example, in Map A, both max throughput and dedicating vehicles achieve dismal fairness. The max throughput schedule only serves customer 1's cluster, and dedicating a vehicle to customer 2 cannot deliver a fair share of throughput, given the round-trip budget constraints. When customers' tasks overlap and have similar spatial density (Map D), the max throughput schedule provides a roughly fair allocation of rates, and the round-robin schedule achieves reasonably high throughput. Dedicating vehicles suffers from poor throughput when there is an incentive to pool tasks from multiple customers into a single vehicle (Map B and Map C).

\begin{figure}
  \centering
  \includegraphics[scale=0.65]{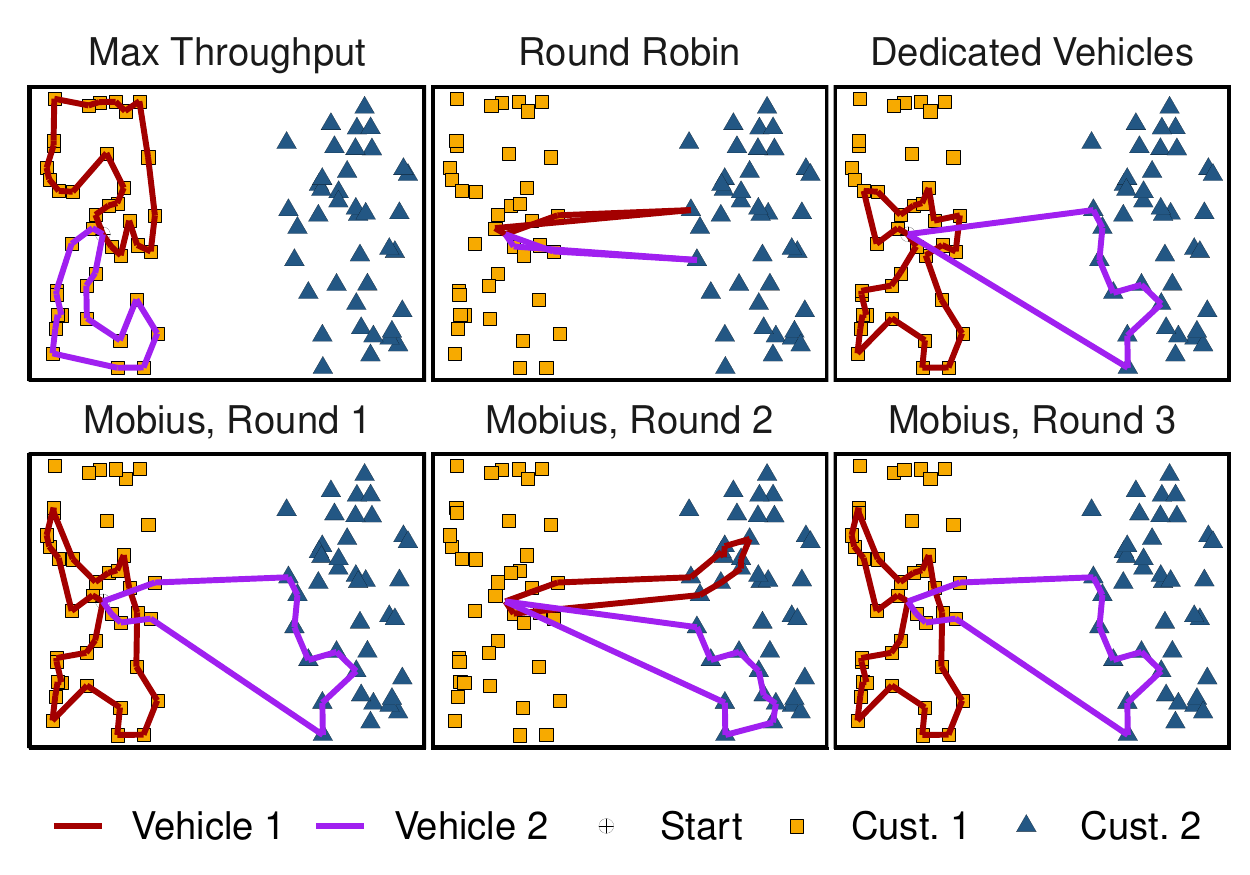}  \vspace{-10pt}
  \caption{Snapshot of per-round schedules computed by \name (for 3 rounds) and other policies. \name compensates for short-term unfairness by switching between schedules on the convex hull over rounds. Other schemes suffer from persistent bias or low throughput.}
  \vspace{-10pt}
  \label{fig:eval:synthetic-routes}
\end{figure}

\vspace{5pt}
\paragraphb{Focusing on Map A.}
To illustrate how \name converges to fair per-customer allocations without significantly degrading platform throughput, we show in \fig~\ref{fig:eval:synthetic-routes} schedules computed for Map A (\fig~\ref{fig:eval:synthetic-cust2-map}). On the top row, we show max throughput, round-robin, and dedicated schedules; on the bottom row, we show schedules computed by \name over 3 consecutive rounds. In round 1, \name \emph{exploits a sharing incentive} to pool some of customer 1's tasks into the journey to customer 2's tasks, achieving a similar throughput to the max throughput schedule. By contrast, the max throughput schedule starves customer 2, and the round-robin schedule wastes time moving between clusters.

\name is able to compensate for short-term unfairness across multiple rounds of scheduling. \fig~\ref{fig:eval:synthetic-routes} shows the first 3 schedules that \name computed for max-min fairness, assuming fulfilled tasks reappear, as before. Although \name does not starve customer 2 in round 1, it still delivers $4\times$ higher throughput to customer 1. However, as we see in round 2, \name compensates for this unfairness by prioritizing customer 2, while still exploiting sharing incentive and collecting a few tasks for customer 1 during the round trip. The schedule in round 3 is identical to that in round 1; since tasks arrive according to a static model in this example, the convex hull remains the same across rounds, and so \name oscillates between the same two support allocations (schedules). The max throughput and dedicated schedules suffer from a persistent bias in throughput (\fig~\ref{fig:eval:synthetic-cust2-sum}) due to the skew in spatial demand.

\subsection{Expressive Schedules with \texorpdfstring{$\alpha$}{a}}
\name's  $\alpha$ parameter allows the platform operator to control fairness; with higher $\alpha$, the platform trades off some total throughput for a fairer allocation of per-customer rates. \fig~\ref{fig:eval:synthetic-cust2-thp} shows, for three different values of $\alpha$, the long-term throughput for each customer and the platform throughput over time. For all maps (\fig~\ref{fig:eval:synthetic-cust2-map}), as we increase the degree of fairness $\alpha$, \name compromises some platform throughput. \name's scheduler is \emph{expressive}; for instance, if an operator would like high throughput, with the only constraint that no customer gets starved, she can run \name with $\alpha = 1$, which ensures a proportionally fair allocation of throughputs. Map A shows an example of this. Furthermore, \name indeed \emph{converges to the target throughput} (\S\ref{sec:design:opt}); this is best illustrated with the max-min schedules, where the customer throughputs converge to the same rate. \name can also converge to any allocation of rates in the spectrum between maximum throughput and max-min fairness (\eg $\alpha = 10$).

\begin{figure}
  \centering
  \includegraphics[scale=0.65]{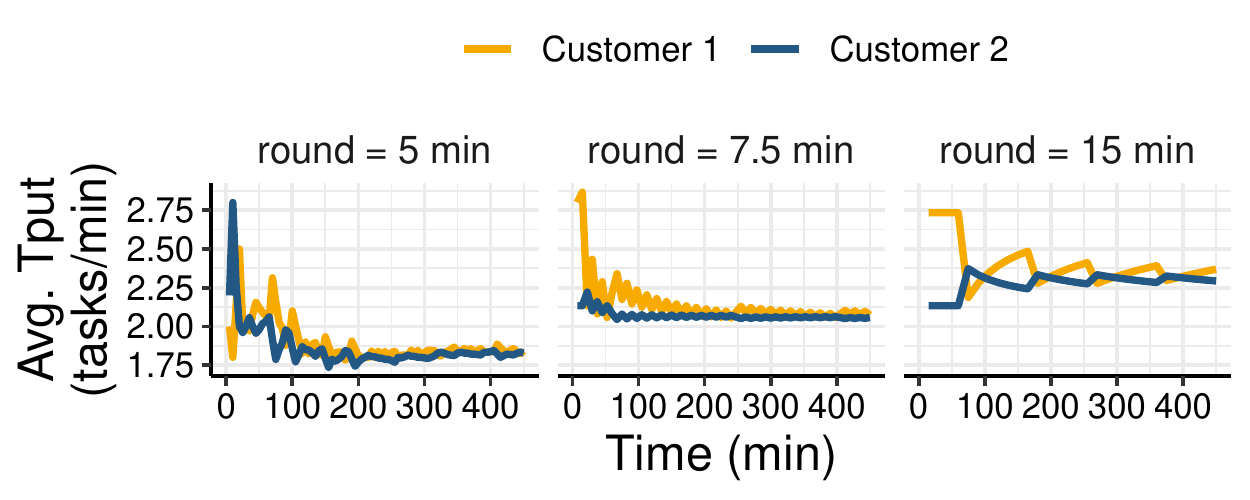}  \vspace{-10pt}
  \caption{\name converges to the fair allocation of throughputs regardless of the timescale of fairness. Scheduling in shorter rounds converges faster to the fair allocation of rates, but longer round durations lead to schedules with greater platform throughput.}
  \vspace{-10pt}
  \label{fig:eval:synthetic-timescale}
\end{figure}

\subsection{Timescale of Fairness}
The duration of a round in \name is a parameter; for instance, an operator could set the round length to be the vehicles' fuel time, or the desired timescale of fairness. We expect that, with shorter durations, \name can converge faster to a fair allocation. To study this behavior, we consider, in \fig~\ref{fig:eval:synthetic-timescale}, max-min fair schedules generated by \name on Map A. We consider three round durations (5 mins, 7.5 mins, and 15 mins), requiring the vehicles to return home every 15 minutes, as before. There are three interesting takeways. First, for all round durations, \name provides an equal allocation of rates to both customers. Second, we see that \name achieves lower platform throughput for shorter round durations; this is because schedules computed at shorter timescales are more myopic. Third, shorter round durations allow \name to converge to faster to an equal allocation rates. In particular, we see that \name at a 5-min timescale achieves the fair allocation within 150 minutes, but at a 15-min timescale, it takes nearly 400 minutes to converge. Thus, the timescale of fairness of fairness allows an operator to trade some total platform throughput for faster convergence.

Additionally, the schedules generated with 5-min and 7.5 min timescales do not observe the static task arrival assumption (\S\ref{sec:design:opt}). Since the vehicles begin some rounds away from their start locations, the convex hull changes across rounds. Still, \name is robust in this setting and provides very similar rates to both customers. The case studies (\S\ref{sec:eval:lyft}-\S\ref{sec:eval:drone}) provide more realistic examples where the static task arrival assumption is relaxed.

\subsection{Geometry of the Convex Boundary}
\label{app:results:boundary}
\begin{figure}
  \centering
  \includegraphics[scale=0.5]{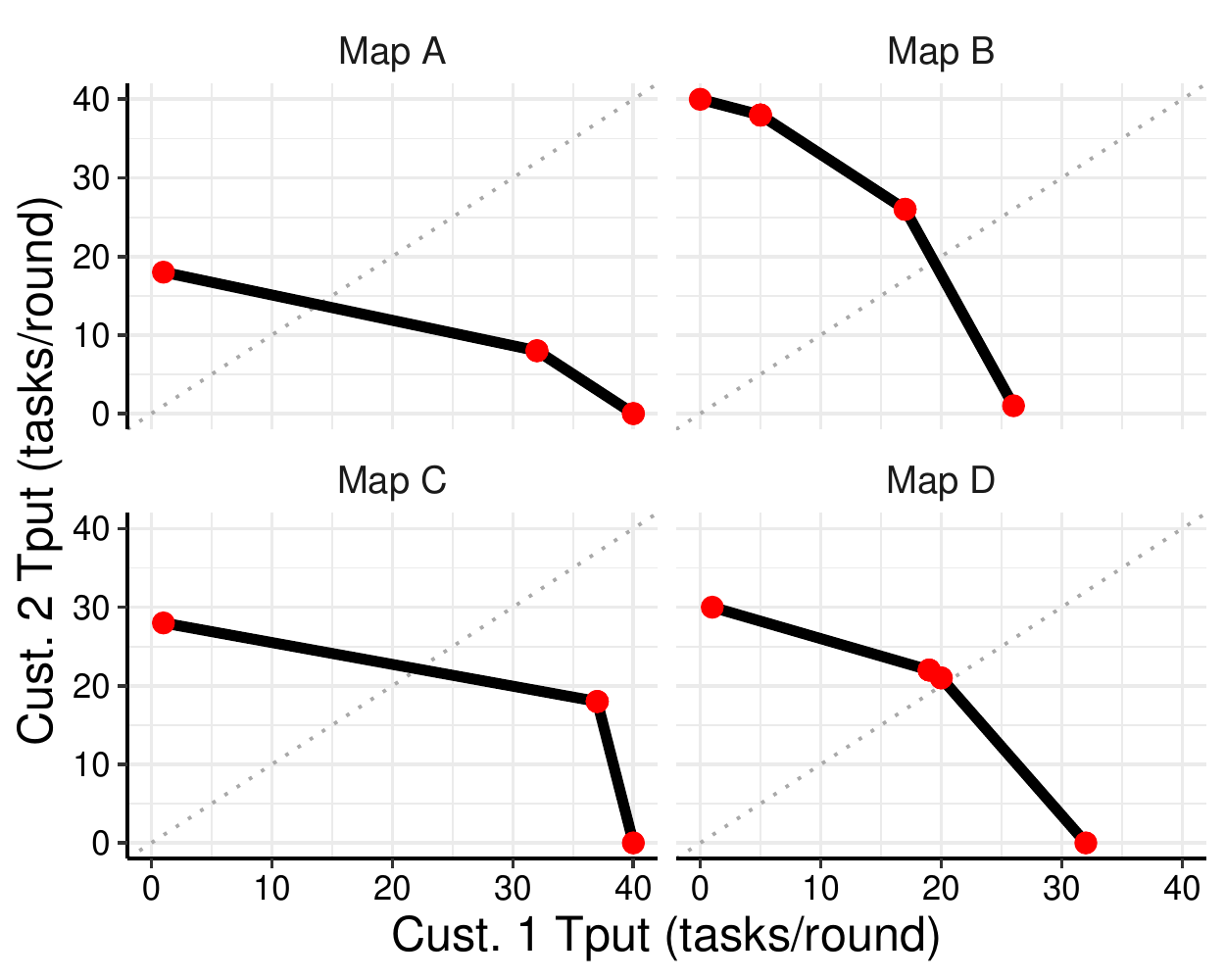}
  \vspace{-10pt}
  \caption{Convex boundaries computed by \name for the different maps shown in \fig~\ref{fig:eval:synthetic-cust2-map}. The shape of the convex boundary describes the inherent tradeoff between fairness and high throughput.}
  \vspace{-10pt}
  \label{fig:app:boundary}
\end{figure}

\name finds an approximately fair schedule in each round by constructing corner points of the convex boundary. The convex boundary of achievable throughputs succinctly captures the tradeoffs between servicing different customers, based on their spatial demand. \fig~\ref{fig:app:boundary} shows the convex boundaries for four different maps of tasks (shown in \fig~\ref{fig:eval:synthetic-cust2-map}). We construct the convex boundary using an extended version of \name's search boundary. Specifically instead of searching the face containing the utility optimum on the \emph{current} convex boundary, we search \emph{all} faces, \ie extend the convex boundary in all directions. The terminating conditions remain the same: we know we have reached face on the convex boundary when we cannot extend it further.

The geometry of the convex boundary indicates which customers the platform can service more easily. For instance, notice that the boundaries for Map A and Map C are both skewed toward customer 1, since the platform incurs less overhead to service both customers. In contrast, the convex boundary for Map D is symmetric, since each customer is equally easy to service.

\subsection{Varying Number of Vehicles}
\label{app:results:vehicles}

The results in \S\ref{sec:eval} show that dedicating vehicles can (i) miss out on sharing incentive, leading to lower platform throughput, and (ii) lead to unfairness in situations where it is inherently harder to service some customers. However, dedicating vehicles is only a viable policy when the number of vehicles is a multiple of the number of customers. \name makes no assumptions about the number of vehicles in the \maas platform; in this section, we study the performance of \name in a \maas platform with different numbers of vehicles.

\begin{figure}
  \centering
  \includegraphics[scale=0.65]{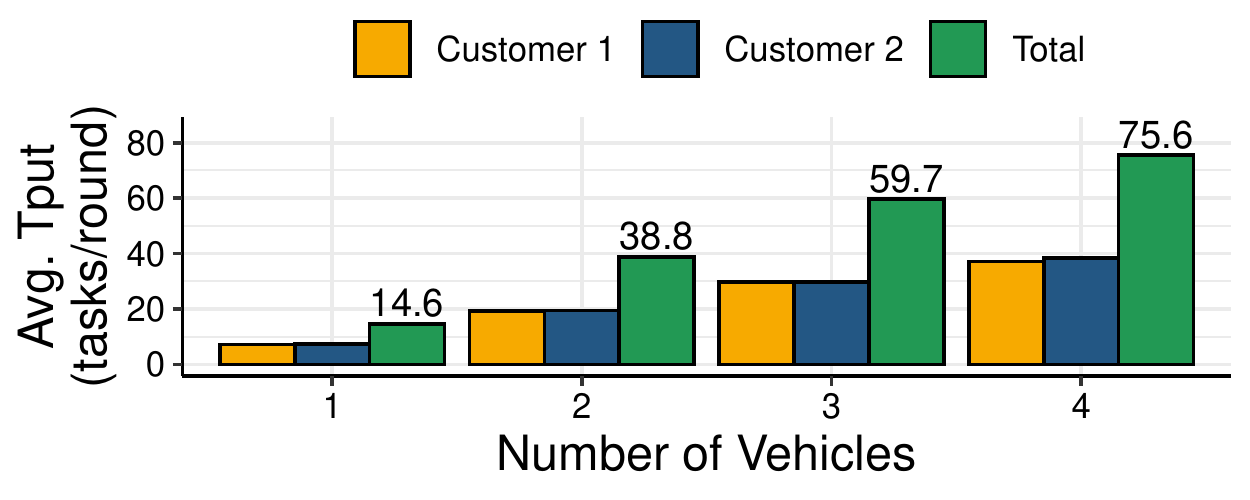}
  \vspace{-10pt}
  \caption{Long-term per-customer rates computed by \name on Map D in \fig~\ref{fig:eval:synthetic-cust2-map}, for different provisioning of vehicles.}
  \vspace{-10pt}
  \label{fig:app:vehicles}
\end{figure}

\fig~\ref{fig:app:vehicles} shows the per-customer average throughputs $\overline{x}_k(t)$ achieved by \name (for max-min fairness) on Map B (\fig~\ref{fig:eval:synthetic-cust2-map}) for different numbers of vehicles. In all cases, \name converges to a max-min fair allocation of rates. As expected, the throughput of the platform increases with more vehicles, since the platform can complete more in parallel.

\subsection{A Case with Three Customers}
\label{app:results:three}

\begin{figure}
  \centering
  \includegraphics[scale=0.65]{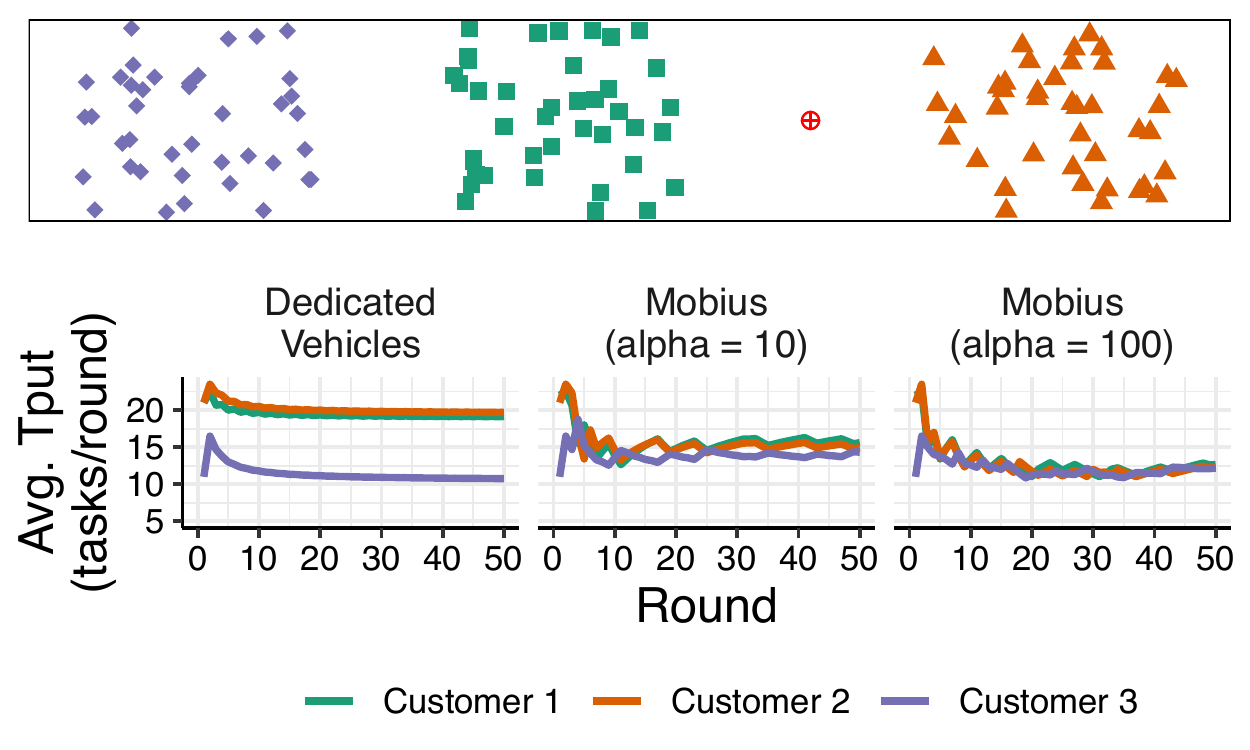}
  \vspace{-10pt}
  \caption{\name vs. dedicating vehicles for example with 3 customers. \name converges to a fair allocation of rates for customers, when the assumption on static task arrival is relaxed.}
  \vspace{-10pt}
  \label{fig:app:three}
\end{figure}

\S\ref{app:eval:synthetic} showed a controlled study of the properties of \name, in environments with two customers. \fig~\ref{fig:app:three} shows an example with three customers and three vehicles, all starting at $\oplus$. We let customers renew fulfilled tasks after every round. We consider a fairness timescale of 5 minutes, and require that the vehicles return home every 15 minutes (\ie 3 rounds); so we relax the assumption on static task arrival, \ie the convex boundary is identical every 3 rounds. \fig~\ref{fig:app:three} shows time series chart of per-customer long-term throughputs achieved by \name (for $\alpha = 10$ and $\alpha = 100$) and by dedicating vehicles. The schedule that dedicates vehicles to customers misses out on the opportunity to fulfill tasks for customer 1 on the way to customer 3's cluster. Additionally, notice that \name can provide a fair allocation of rates for 3 customers, and $\alpha$ allows \name to control the degree to which the rates converge to the same value.

    \fi
\end{sloppypar}
\end{document}